\documentclass[11pt]{article}
\usepackage{amsmath,amssymb,amsthm}
\usepackage{titlesec}

\usepackage[pdftex]{graphicx}
\usepackage{soul}
\usepackage[margin=1in]{geometry}
\usepackage[normalem]{ulem}
\usepackage{color}
\usepackage{xcolor}
\usepackage[hidelinks]{hyperref} 
\usepackage{pbox}
\usepackage{verbatim} 
\usepackage{multicol}
\usepackage{bm}

\newtheorem{theorem}{Theorem}[section]

\newtheorem{corollary}[theorem]{Corollary}

\newtheorem{definition}[theorem]{Definition}

\newtheorem{lemma}[theorem]{Lemma}

\newtheorem{proposition}[theorem]{Proposition}
\newtheorem{remark}[theorem]{Remark}

\DeclareMathOperator{\vol}{vol}

\DeclareMathOperator{\tr}{tr}

\DeclareMathOperator{\Sym}{Sym}
\DeclareMathOperator{\PSD}{PSD}
\DeclareMathOperator{\PD}{PD}
\DeclareMathOperator{\supp}{supp}
\DeclareMathOperator*{\argsup}{arg\,sup}

\DeclareMathOperator{\cov}{cov}

\newcommand{\R}{\mathbb{R}}
\newcommand{\Z}{\mathbb{Z}}
\newcommand{\N}{\mathbb{N}}

\newcommand{\diff}{\mathrm{d}}

\newcommand{\be}{\begin{equation}}
\newcommand{\ee}{\end{equation}}
\newcommand{\bea}{\begin{eqnarray}}
\newcommand{\eea}{\end{eqnarray}}

\newcommand{\bc}{\begin{center}}
\newcommand{\ec}{\end{center}}
\newcommand{\ben}{\begin{enumerate}}
\newcommand{\een}{\end{enumerate}}
\newcommand{\bfi}{\begin{figure}}
\newcommand{\efi}{\end{figure}}

\title{Deterministic Approximation Algorithms for Volumes of Spectrahedra}
\author{Mahmut Levent Do\u{g}an\thanks{Technische Universit\"at Berlin, Institut
    f\"ur Mathematik, Strasse des 17. Juni 136, 10623, Berlin,
    Germany
    (dogan@math.tu-berlin.de)}\and Jonathan Leake\thanks{University of Waterloo, Department of Combinatorics and Optimization, 200 University Ave W, Waterloo, ON, Canada (jonathan.leake@uwaterloo.ca)} \and Mohan Ravichandran\thanks{Department of Mathematics, Bogazici University, Bebek, Istanbul (mohan.ravichandran@gmail.com).}}

\begin{document}

\maketitle

\begin{abstract}
    We give a method for computing asymptotic formulas and approximations for the volumes of spectrahedra, based on the maximum-entropy principle from statistical physics.
    %
   %
    %
    The method gives an approximate volume formula based on a single convex optimization problem of minimizing $-\log \det P$ over the spectrahedron.
    Spectrahedra can be described as affine slices of the convex cone of positive semi-definite (PSD) matrices, and the method yields efficient deterministic
    %
    approximation algorithms and asymptotic formulas
    whenever the number of affine constraints is sufficiently dominated by the dimension of the PSD cone.
    %
    %
    
    Our approach is inspired by the work of Barvinok and Hartigan who used an analogous framework for approximately computing volumes of polytopes. Spectrahedra, however, possess a remarkable feature not shared by polytopes, a new fact that we also prove: central sections of the set of density matrices (the quantum version of the simplex) all have asymptotically the same volume. This allows for very general approximation algorithms, which apply to large classes of naturally occurring spectrahedra. 
    
    We give two main applications of this method.
    First, we apply this method to what we call the ``multi-way Birkhoff spectrahedron'' and obtain an explicit asymptotic formula for its volume.
    This spectrahedron is the set of quantum states with maximal entanglement (i.e., the quantum states having univariant quantum marginals equal to the identity matrix) and is the quantum analog of the multi-way Birkhoff polytope.
    Second, we apply this method to explicitly compute the asymptotic volume of central sections of the set of density matrices.
    %
    %
    
    %
\end{abstract}

\clearpage

\tableofcontents

\clearpage
\pagenumbering{arabic}

\section{Introduction}


Approximate computation of the volume of convex sets is a fundamental problem in computer science. Although efficient deterministic approximation is impossible in general (e.g., see \cite{dyerfrieze, elekes_geometric_1986}), a large number of efficient randomized algorithms exist. The first such algorithm was given in \cite{kannanetal}, and since then many improvements and other algorithms have been given \cite{lovasz1993random,betke1993approximating,kannan1997random,lovasz2006simulated,cousins2016practical,cousins2017efficient,chalkis2019practical}. Generally speaking, the problem of volume computation of convex sets has generated a large amount of literature in mathematics and computer science.


This work deals with the more specific problem of computing the volume of \emph{spectrahedra}, an important class of convex sets that are well-studied in real algebraic geometry, optimization, and beyond (e.g., see \cite{todd_2001,blekherman2012semidefinite,wolkowicz2012handbook, alizadeh} and the references therein). Spectrahedra are affine slices of the cone of real symmetric positive semi-definite (PSD) matrices, and in that sense, they are relatives of polytopes, which are affine slices of the positive orthant. However, unlike the class of polytopes where there are a number of novel approaches with different practical and theoretical advantages \cite{emiris2014efficient,emiris2018practical,chen2018fast,lee2018convergence,mangoubi2019faster,chalkis2020volesti,barvinok2021quick}, volume computation for spectrahedra is often performed using vanilla random walk methods developed for general convex sets without regard to the explicit presentations that are available \cite{elias}. In this work, we revisit an idea due to Barvinok and Hartigan \cite{bhgaussian} for deterministically approximating the volume of polytopes, and we adapt this technique to spectrahedra.

In \cite{bhgaussian}, Barvinok and Hartigan employ the maximum entropy approach to approximate the volume of a polytope $\mathcal{P}$, represented as an affine slice of the positive orthant. This type of approach was originally used by Jaynes \cite{jaynes1957informationi,jaynes1957informationii}, who was motivated by problems in statistical mechanics.
The general idea is to estimate the average value of a given functional $f$ over an unknown distribution $\mu$ by assuming the distribution maximizes entropy, conditioned on some given assumptions. In \cite{bhgaussian}, the unknown distribution $\mu$ is the uniform distribution on $\mathcal{P}$, and the functional $f$ is the density function of $\mu$, which is inversely proportional to the volume of $\mathcal{P}$. (When the polytope $\mathcal{P}$ is sufficiently complicated, it is reasonable to consider $\mu$ to be unknown.) By assuming $\mu$ to be entropy-maximizing on the positive orthant and proving a local central limit theorem, Barvinok and Hartigan are able to approximate the density function of $\mu$ by the density function of an appropriate Gaussian at its expectation. Under certain conditions on the constraints and the dimension of the polytope, this gives rise to a deterministic approximate volume formula for $\mathcal{P}$ which depends on the input parameters and the solution to a simple convex optimization problem over $\mathcal{P}$.

In this paper, we adapt this method to obtain an asymptotic approximation algorithm for the volume of a given family of spectrahedra $\{\mathcal{S}_n\}_{n=1}^\infty$ whenever the number of affine constraints which define $\mathcal{S}_n$ is sufficiently dominated by the dimension of the PSD cone in which $\mathcal{S}_n$ lies. An \emph{asymptotic approximation algorithm} is an algorithm for approximating a family of values $V(n)$. Given any small $\epsilon$, the algorithm approximates the value $V(n)$ with relative error $\epsilon$ for any given $n \geq n_\epsilon$, where $n_\epsilon$ has nice dependence on $\epsilon$ and possibly other parameters. That is, an asymptotic approximation algorithm is an approximation algorithm where the allowed input size is lower-bounded by a function dependent on the error parameter $\epsilon$.

The asymptotic approximation algorithm we obtain in this paper for $V(n) = \vol(\mathcal{S}_n)$ is essentially a formula, depending on the input parameters and on the solution to a simple convex optimization problem over the given spectrahedron $\mathcal{S}_n$. Thus under certain nice circumstances, our algorithm naturally becomes an explicit asymptotic formula for the given family of spectrahedra. This is similar to the situation in \cite{bhgaussian} for polytopes.

Beyond our general results, we obtain explicit asymptotic formulas for two particular families of spectrahedra (among a number of other examples). First, we provide an asymptotic formula for the real symmetric version of the spectrahedron of mixed quantum states with prescribed marginals. A \textit{quantum state} is defined as a positive semi-definite linear operator $A$ (called the \textit{density matrix}) on the space $\mathbb{C}^{(n_1,\ldots,n_k)}:=\bigotimes_{i=1}^k \mathbb{C}^{n_i}$ with $\tr(A)=1$. For a subset $I\subset \{1,\ldots,k\}$, the \textit{quantum marginal} (or the \textit{reduced density matrix}) $\rho_I(A)$ of $A$ is a linear operator on $\mathbb{C}^{I}:=\bigotimes_{i\in I}\mathbb{C}^{n_i}$ which is obtained by tracing out the spaces $\mathbb{C}^{n_j}$ for all $j \not\in I$. The quantum marginal problem then asks whether a given set of quantum marginals is \textit{consistent}, i.e., given a collection of reduced density matrices does there exist a quantum state $A$ with these quantum marginals? For a study of the quantum marginal problem we refer to \cite{kly1, kly2}. For recent developments that demonstrate the connection between the quantum marginal problem and geometric complexity theory, see \cite{altmin, burgisser2018efficient}. 

A special case of the quantum marginal problem is the study of univariant marginals, i.e., the quantum marginals obtained by considering $\rho_{\{i\}}(A)$ for all $i \in \{1,\ldots,k\}$. These quantum marginals are always consistent and the set of quantum states with given fixed quantum marginals form a non-empty convex set. In fact, such sets of quantum states are affine slices of the positive semi-definite cone, and thus they are spectrahedra. Our method applies to the the real symmetric versions of these families of spectrahedra, and we give a formula for the asymptotic volume in the special case when the univariant marginals are all identity matrices in Corollary \ref{cor:multi-way-vol}. These spectrahedra can be viewed as the quantum analog of the multi-way Birkhoff polytope, consisting of all multi-stochastic completely positive maps.

As a second example, we provide an asymptotic formula for central sections of the real symmetric standard spectraplex. The \emph{standard spectraplex} $\mathcal{S}_N$ is the spectrahedron consisting of all $N \times N$ PSD matrices with trace equal to 1 (called \emph{density matrices}), and it can be considered as the quantum analog of the standard simplex. A central section is then any intersection of the spectraplex with a codimension-one affine hyperplane passing through $\frac{1}{N} I_N$, the center of $\mathcal{S}_N$. In Corollary \ref{cor:sections}, we show a remarkable fact: every sequence of central sections of $\mathcal{S}$ of increasing dimension has the same asymptotic formula. This implies the ratio of the maximum and minimum volume of central sections of $\mathcal{S}_N$ approaches 1 as $N \to \infty$. This differs from the case of the standard simplex, where this ratio is bounded below by a constant greater than 1, see \cite{webb, brzezinski2013volume}.






\section{Main Results} \label{sec:main_result}

Define $\Sym(N)$ to be the space of real symmetric $N \times N$ matrices, $\PSD(N)$ to be the set of real symmetric positive semi-definite matrices, and $\PD(N)$ to be the set of real symmetric positive definite matrices.
Throughout, we consider $\Sym(N)$ as a real inner product space with Frobenius inner product defined via $\langle X, Y \rangle_F := \tr(XY)$. Given $A_1,A_2,\dots,A_m\in\Sym(N)$ and $b\in\R^m$, we define a spectrahedron $\mathcal{S}$ by:
\[
    \mathcal{S} := \big\{P \in \PSD(N) \; : \; \tr(A_k P) = b_k \quad \text{for} \quad k \in [m]\big\}.
\]
We assume that $\mathcal{S}$ is compact, that the constraints $\tr(A_k P) = b_k$ are linearly independent, that $m < \binom{N+1}{2} = \dim(\PSD(N))$, and that $\mathcal{S}$ is of dimension exactly $\binom{N+1}{2} - m$.
We now present our main results, leaving the proofs to Section \ref{sec:proofs_corollaries}.

\subsection{General approximation and asymptotic results}

Given a spectrahedron $\mathcal{S}$ as defined above, let $P^\star \in \mathcal{S}$ be the point which maximizes the function 
\[
\begin{split}
    \phi(P) &= \log\Gamma_N\left(\frac{N+1}{2}\right) - \frac{N(N+1)}{2} \log\left(\frac{N+1}{2e}\right) + \frac{N+1}{2}\log\det(P) \\
        &= \text{const}(N) + \frac{N+1}{2}\log\det(P)
\end{split}
\]
over $\mathcal{S}$, where $\Gamma_N$ is the multivariate gamma function. The function $\phi$ is the entropy function of an associated Wishart distribution on $\PSD(N)$; see Definition \ref{def:entropy_function} and Corollary \ref{cor:entropy_function_formula}. We discuss this further in Section \ref{sec:max-entropy}.
Let $A$ and $B$ be linear operators from $\Sym(N)$ to $\R^m$, defined via
\[
    AX := (\tr(A_1X), \ldots, \tr(A_mX))
\]
and
\[
    BX := \left(\tr\left(\sqrt{P^\star}A_1\sqrt{P^\star}X),\ldots,\tr(\sqrt{P^\star}A_m\sqrt{P^\star}X\right)\right).
\]
In the following results, we approximate the volume of $\mathcal{S}$ by the formula
\[
    \vol(\mathcal{S}) \approx \left(\frac{N+1}{4\pi}\right)^{m/2} \left(\frac{\det(AA^\top)}{\det(BB^\top)}\right)^{1/2} e^{\phi(P^\star)},
\]
under certain conditions on $N$ and $m$. Our first result gives conditions under which the formula yields a good approximation for $\vol(\mathcal{S})$.

\begin{theorem}[Main approximation result] \label{thm:approx}
    
    
    Let $\mathcal{S}$ be a spectrahedron defined with the above notation. Fix $\epsilon \leq e^{-1}$ and suppose that
    \begin{align} \label{eq:main4}
        \frac{\epsilon^2}{\log^3(\epsilon^{-1})} \geq \frac{\gamma \, m^3 \log N}{N},
    \end{align}
    where $\gamma$ is an absolute constant (we can choose $\gamma = 32 \cdot 10^5$). Then the number
    \[
        \left(\frac{N+1}{4\pi}\right)^{m/2} \left(\frac{\det(AA^\top)}{\det(BB^\top)}\right)^{1/2} e^{\phi(P^\star)}
    \]
    approximates $\vol(\mathcal{S})$ within relative error $\epsilon$.
\end{theorem}
    

Our second result gives conditions under which the formula yields the correct asymptotics for the volume of a family of spectrahedra.

\begin{theorem}[Main asymptotic result] \label{thm:asymptotic}
    Let $\{\mathcal{S}_n\}_{n=1}^\infty$ be a family of spectrahedra with the above notation, using the subscript $n$ to denote which spectrahedron in the family $\{\mathcal{S}_n\}_{n=1}^\infty$ we are referring to. If
    \begin{align} \label{eq:main5}
        \lim_{n \to \infty} \left[\frac{m_n^3 \log N_n}{N_n}\right] = 0,
    \end{align}
    then
    \[
        \lim_{n \to \infty} \frac{\vol(\mathcal{S}_n)}{\left(\frac{N_n+1}{4\pi}\right)^{m_n/2} \left(\frac{\det(A_nA_n^\top)}{\det(B_nB_n^\top)}\right)^{1/2} e^{\phi_n(P_n^\star)}} = 1.
    \]
    That is, we achieve an asymptotic formula (depending on $P_n^\star$) for the volume of $\mathcal{S}_n$.
\end{theorem}

Utilizing the notation of Theorem \ref{thm:asymptotic} for a family of spectrahedra $\{\mathcal{S}_n\}_{n=1}^\infty$, the above results imply an asymptotic approximation algorithm whenever Condition (\ref{eq:main5}) is satisfied: Given small $\epsilon > 0$, there exists $n_\epsilon$ such that for all $n \geq n_\epsilon$ we have a formula (based on the optimizer $P_n^\star$) which approximates $\vol(\mathcal{S}_n)$ within relative error $\epsilon$. Computing the optimizer $P_n^\star$ of the convex function $-\log \det P$ is the last required step of the algorithm, and this can be done efficiently using standard convex optimization techniques like interior point methods or the ellipsoid method. See \cite{maxdet} for further discussion related to this specific optimization problem.

\subsection{The multi-way Birkhoff spectrahedron}

We now give a formula for the asymptotic volume of the multi-way Birkhoff spectrahedron, consisting of all multi-stochastic completely positive maps on real symmetric matrices. This spectrahedron can be viewed as the quantum analog of the multi-way Birkhoff polytope, but we leave the formal definitions of this and all notation to Section \ref{sec:multi-CP-maps}. We also give the proof of the asymptotic formula in Section \ref{sec:multi-CP-maps}.
Fix $n,k\in\mathbb{N}$ and consider the spectraplex \[
\mathcal{S} := \{ A\in\PSD(n^k)\, :\, \tr(A)=1 \},
\] which can be interpreted as the set of density matrices (with real entries), acting on the space $(\R^n)^{\otimes k}$. For $i=1,2,\dots,k$, the $i$-th \textit{partial trace} operator is the unique linear map \[
\tr_i : \PSD(n^k) \rightarrow \PSD(n)
\] defined by the property that \[
\tr_i( A_1\otimes A_2\otimes\dots\otimes A_k ) = \tr(A_{\widehat{\imath}})\, A_i,
\] where $A_{\widehat{\imath}}$ is interpreted as a linear operator on $(\R^n)^{\otimes (k-1)}$. For a density matrix $A$, the partial traces $\tr_i(A)$ are called the (univariant) \textit{quantum marginals} of $A$. It is known that $\tr_i$ is a completely positive map and $\tr_i(A)\in\PSD(n)$ for every $A\in\mathcal{S}$. The \textit{multi-way Birkhoff spectrahedron} is then defined to be \[
\mathcal{SCP}_{n,k} := \left\{ A\in\PSD(n^k)\, :\, \tr(A) = n \text{ and } \tr_i(A)= I_n \text{ for all } i = 1, \ldots, k \right\}.
\] 

\begin{corollary} \label{cor:multi-way-vol}
    Fix $k \geq 7$, and set $N := n^k$ and $m := k \binom{n+1}{2} - k + 1$ for any $n \in \N$. As a function of $n$, the asymptotic volume of $\mathcal{SCP}_{n,k}$ is given by
    \[
        \vol(\mathcal{SCP}_{n,k}) \approx \left(\frac{N+1}{4\pi}\right)^{\frac{m}{2}} \left(\frac{N}{n}\right)^m \left(\frac{2en}{N(N+1)}\right)^{\frac{N(N+1)}{2}} \Gamma_N\left(\frac{N+1}{2}\right).
    \]
\end{corollary}

Now assume that $k=2$ and consider the intersection of the Birkhoff spectrahedron with the set of $n^2\times n^2$-diagonal matrices. Then, $\mathcal{SCP}_{n,k}\cap\mathrm{Diag}(n^2)$ is a polytope. Indeed, this polytope is the well known \textit{Birkhoff polytope} of $n\times n$, doubly stochtastic matrices. Similarly, for $k\geq 3$, the intersection of $\mathcal{SCP}_{n,k}$ with diagonal matrices gives rise to a polytope, the multi-way Birkhoff polytope. This observation justifies the name \textit{Birkhoff spectrahedron} and we view $\mathcal{SCP}_{n,k}$ as the quantum analog of the Birkhoff polytope.

\begin{remark}
    The usual definitions of quantum states and quantum marginals consider operators on complex vector spaces rather than real ones. Our general results above only apply to the real symmetric PSD cone, and thus we only give the asymptotic formulas in the real symmetric case. That said, we believe similar results should be possible for the Hermitian PSD cone, and even the quaternionic PSD cone; see the discussion at the end of Section \ref{sec:technical-to-main}.
\end{remark}

\subsection{Central sections of the standard spectraplex}

Now consider the \emph{standard spectraplex}, defined via
\[
    \mathcal{S}_1 := \left\{P \in \PSD(N)\, :\, \tr(P) = 1\right\}.
\]
Equivalently $\mathcal{S}_1$ is the set of real symmetric density matrices, and it is the spectrahedral analog of the standard simplex. Both $\mathcal{S}_1$ and $\phi(P)$ are invariant under conjugation by invertible matrices, which shows that the analytic center of $\mathcal{S}_1$ equals $P^* = \frac{1}{N}I_N$. A \emph{central section} of $\mathcal{S}_1$ is then given by intersecting the spectraplex with an affine hyperplane through its center $\frac{1}{N} I_N$. That is,
\[
    \mathcal{S}_M := \left\{P \in \PSD(N)\, : \,\tr(P) = 1 \text{ and } \tr(MP) = \frac{\tr(M)}{N}\right\},
\]
for some $M \in \Sym(N)$ linearly independent of $I_N$.
Note that $P^*=\frac{1}{N}I_N$ always satisfies the second condition, and thus the analytic center of $\mathcal{S}_M$ is $\frac{1}{N}I_N$. Therefore the approximation formula from Theorem~\ref{thm:approx} gives the same formula for every choice of $M$. Our next result shows that, asymptotically, this formula actually gives a good approximation for $\vol(\mathcal{S}_M)$. This result is an immediate corollary of Theorem \ref{thm:asymptotic}.


\begin{corollary}[Central sections of the spectraplex]
\label{cor:sections}
    
    For every $\epsilon < e^{-1}$, there exists $N_\epsilon$ such that for all $N \geq N_\epsilon$ and all $M \in \Sym(N)$ linearly independent of $I_N$, the number
    \[
        \frac{N^2(N+1)}{4\pi} \, e^{\phi(\frac{1}{N} I_N)} = \frac{N^2(N+1)}{4\pi} \left(\frac{2e}{N(N+1)}\right)^{\frac{N(N+1)}{2}} \Gamma_N\left(\frac{N+1}{2}\right)
    \]
    approximates $\vol(\mathcal{S}_M)$ within relative error $\epsilon$. In particular, there is a single asymptotic formula for the volume, which holds for any sequence of central sections of increasing dimension. 
\end{corollary}



Corollary~\ref{cor:sections} demonstrates a striking distinction between spectrahedra and polytopes: All central sections of the spectraplex have the same asymptotic volume, but the analogous statement is false for the standard simplex. The analytic center of the standard simplex $\Delta_n = \{ x\in\R^{n+1}_{\geq 0}\, : \, \sum_{i=1}^{n+1} x_i = 1\}$ is the scaled all-ones vector $\frac{1}{n+1}\mathbf{1}_{n+1}$, and there exist two central sections of the standard simplex such that the ratio of their volumes limits to $\frac{e}{\sqrt{2}}$ as $n\to\infty$ (one containing $n-1$ vertices of $\Delta_n$, and the other one parallel to one of the facets). It is an open problem whether this ratio is optimal. We refer to \cite{webb,brzezinski2013volume,aubrun2017alice} for relevant computations and further discussion.

\begin{remark}
    A version of Corollary \ref{cor:sections} actually holds for spectrahedra obtained by intersecting $\mathcal{S}_1$ with any collection of $m$ linearly independent affine hyperplanes that pass through $\frac{1}{N}I_N$, where $m$ is a fixed constant. One can essentially obtain the same result with the same proof; the only difference is that the approximate volume formula will depend on $m$ (but not on the choice of affine hyperplanes).
\end{remark}

\section{Technical Overview}

In this section, we give an overview of the techniques used for the main results.
The main results given in Section \ref{sec:main_result} are all derived in Section \ref{sec:proofs_corollaries} from our main technical result given in Theorem \ref{thm:main}.
That said, we first give a birds-eye view on the overarching conceptual strategy of the proof of our main technical result, and then we discuss in more detail what must actually be done to make this strategy work.
After that, we briefly discuss how we use our main technical result (Theorem \ref{thm:main}) to prove our main results.
%

\subsection{The overarching strategy}

In \cite{bhgaussian}, the maximum-entropy approach is utilized to approximate the volume of certain polytopes.
(This approach is also utilized to approximately count integer points of polytopes, but we will not discuss this further here.)
The authors consider a polytope $\mathcal{P}$ to be given by an affine slice of the positive orthant, denoted by
\[
    \mathcal{P} := \{x \in \R_{\geq 0}^n : Ax = b\}
\]
for some fixed $m \times n$ real matrix $A$ and $b \in \R^m$, where $m < n$ and $A$ has rank $m$.
%
%
One main goal of their paper is then to approximate the volume of $\mathcal{P}$ in various situations.

In this paper, we adapt the maximum-entropy strategy from \cite{bhgaussian} to the case of spectrahedra.
We consider a spectrahedron $\mathcal{S}$ to be given by an affine slice of the positive semi-definite (PSD) cone contained in the space of real symmetric matrices, denoted by
\[
    \mathcal{S} := \{P \in \PSD(N) : \tr(A_k P) = b_k \text{ for } k = 1,2,\ldots,m\}
\]
for some fixed $N \times N$ real symmetric $A_1,\ldots,A_m$ and $b \in \R^m$, where $m < \binom{N+1}{2} = \dim(\PSD(N))$ and $A_1,\ldots,A_m$ are linearly independent. We will denote by $A$ the linear map from $N\times N$ matrices to $\R^m$ sending $P$ to $(\tr(A_1P),\tr(A_2P),\dots,\tr(A_mP))$, and thus $\mathcal{S}$ can also be written as $\mathcal{S}=\{P \in \PSD(N)\,:\, AP=b\}$.
We assume that $\mathcal{S}$ is a compact subset of the PSD cone.

Our main goal is then to approximate the volume of $\mathcal{S}$ in various situations.
The overarching strategy consists of the following steps, which have been adapted from the polytope case in \cite{bhgaussian}.
(One can see the polytope case in the following steps by using the positive orthant $\R_{\geq 0}^n$ instead of the PSD cone.)
\begin{enumerate}
    \item Construct a certain random variable $X$ on $\PSD(N)$ which has expectation in $\mathcal{S}$, and for which the density function $f_X$ is constant on $\mathcal{S}$. This random variable $X$ is distributed according to a certain \emph{maximum-entropy distribution} $\mu$ on $\PSD(N)$, and it has log-linear density function $f_X(P) \propto e^{-\tr(M P)}$ for some $M \in \PD(N)$.
    \item Construct the random variable $Y = AX = \left(\tr(A_1X),\ldots,\tr(A_mX)\right) \in \R^m$ which has expectation $\mathbb{E}[AX] = A\mathbb{E}[X]=b$. Since the density function $f_X$ of $X$ is constant on $\mathcal{S}$, the volume of $\mathcal{S}$ can be related to the density function $f_Y$ of $Y$ via the standard change of variables formula for $Y$:
    \[
        f_Y(b) = \int_{A^{-1}(b)} f_X(x) \frac{\diff x}{\sqrt{\det(AA^\top)}} =
        \frac{\vol(\mathcal{S}) \cdot f_X(P_0)}{\sqrt{\det(AA^\top)}},
    \]
    where $P_0$ is any point of $P$.
    \item Compute $f_X(P_0)$ and $f_Y(b)$ and rearrange to compute $\vol(\mathcal{S})$. Since $f_X$ is log-linear, $f_X(P_0)$ is easy to compute directly. Since $Y$ is a linear transformation applied to a maximum-entropy distribution, it approximates a Gaussian random variable under certain conditions. (Gaussians are the maximum-entropy distributions on $\R^m$.)
    %
    %
    Thus apply a local central limit theorem-type argument to approximate $f_Y(y)$ at its expectation $y=b$.
\end{enumerate}
%
%
Conceptually the steps are not very different between the polytope case and the spectrahedron case, and in fact this overarching strategy theoretically may work for affine slices of any convex body.
%
%
There are two key features that then make this strategy practical for spectrahedra.
First, the maximum entropy value in Step 1 is equal to $\log \det M$ (up to constant) for some $M \in \mathcal{S}$.
The fact that this quantity has a nice formula is a crucial observation (see Section \ref{sec:max-entropy}), and it only occurs for certain classes of convex bodies (including polytopes and spectrahedra).
And second, the technical arguments of \cite{bhgaussian} used to prove the local central limit theorem in Step 3 can be adapted to the spectrahedron case.
This adaptation requires substantial work, as can be seen in Section \ref{sec:proof_main}.

\subsection{Step 1: The maximum-entropy distribution}

Maximum-entropy distributions on various subsets of $\R^n$ have been well-studied since the seminal work of Shannon \cite{shannon1948mathematical} and Jaynes \cite{jaynes1957informationi,jaynes1957informationii}.
The most basic max-entropy optimization problem on a convex domain $K \subset \R^n$ can be formulated as:
\vspace{-0.5em}
\[
    f = \argsup_{\substack{f, \text{ density function} \\\supp(f) \subseteq K}} \int f(x) \log \frac{1}{f(x)} \, \diff x,
    \vspace{-0.5em}
\]
where $\int f(x) \log \frac{1}{f(x)} \, \diff x$ is known as the differential entropy of the distribution $\mu$ with density function $f$.
(Note that solutions may or may not exist for the above problem as stated.)

In the case that $K \subseteq \R^n$ is a compact convex set, there is a unique max-entropy distribution $\mu$: the uniform distribution on $K$.
When $K$ is a non-compact convex domain, there is no longer a unique max-entropy distribution, and thus more information must be specified to obtain uniqueness.
One natural way to do this is to restrict the expectation of $\mu$.
%
In particular, if $K$ is the PSD cone $\PSD(N) \subset \R^{\binom{N+1}{2}}$, then there is a unique max-entropy distribution for each choice of expectation in the strictly positive definite cone $\PD(N)$.
More generally, maximizing entropy over a space of distributions with restricted expectation can yield distributions with interesting properties.

In the case of $K = \PSD(N)$, all max-entropy distributions belong to the family of \emph{Wishart distributions} on $\PSD(N)$ with log-linear density function given by $f(P) \propto e^{-\tr(MP)}$, where $M \in \PD(N)$ depends on the desired expectation of the distribution.
This is proven in Section \ref{sec:max-entropy} by computing the dual convex program to the following \emph{infinite}-dimensional maximum entropy program:
\vspace{-0.5em}
\[
    \psi(\mathcal{S}) = \sup_{\substack{f, \text{ density function} \\ \supp(f) \subseteq \PSD(N) \\ \mathbb{E}[f] \in \mathcal{S}}} \int f(x) \log \frac{1}{f(x)} ~ \diff x.
    \vspace{-0.5em}
\]
We show that the dual convex program to this maximum entropy program is given by the \emph{finite}-dimensional optimization problem
\vspace{-0.5em}
\[
    \psi(\mathcal{S}) = \sup_{P \in \mathcal{S}} \phi(P) = \sup_{P \in \mathcal{S}} \left[\text{const}(N) + \frac{N+1}{2} \log\det(P)\right],
\]
where the first equality above is a strong duality result which says that the optimal values of the primal and dual programs are equal.
The function $\phi(P)$ is precisely the entropy of the associated Wishart distribution with expectation $P$.
%

Thus, the value of $P \in \mathcal{S}$ which optimizes the function $\phi$ gives rise to the max-entropy distribution on $\PSD(N)$ with expectation contained in $\mathcal{S}$.
Specifically, it is given by $f(x) \propto e^{-\tr((P^\star)^{-1}\, P)}$ where $P^\star$ is the optimizer of $\phi$.
We prove this in Section \ref{sec:max-entropy}, but it also follows from the Appendix of \cite{leake2020computability} where a max-entropy strong duality theorem is proven in much greater generality.
See also Section 3 of \cite{bhgaussian} for analogous results in the polytope case.

Given a random variable $X$ with associated density function $f_X(x) \propto e^{-\tr((P^\star)^{-1}\, P)}$, the crucial observation is then that this density function is constant on the spectrahedron $\mathcal{S} \subset \PSD(N)$.
Conceptually this is related to the fact that uniform distributions maximize entropy on compact convex domains.
More concretely, this follows from a basic gradient computation using the dual formulation above.
If $P^\star$ optimizes $\phi$, then the following equivalent conditions hold for all $P,Q \in \mathcal{S}$:
\vspace{-0.5em}
\[
    \langle P-Q, \nabla \phi(P^\star) \rangle = \tr((P-Q) \nabla \phi(P^\star)) = 0 \quad \iff \quad \tr((P^\star)^{-1} P) = \tr((P^\star)^{-1} Q),
    \vspace{0.25em}
\]
since $\nabla \phi(P^\star) = (P^\star)^{-1}$.
Thus $f_X(x) \propto e^{-\tr((P^\star)^{-1} P)}$ is constant on $\mathcal{S}$.

With this, we have a max-entropy distribution on $\PSD(N)$ which is constant on $\mathcal{S}$.
This will enable us to relate the volume of $\mathcal{S}$ to other quantities as described in the overarching strategy above.
Further, computing this max-entropy distribution boils down to a particular finite-dimensional optimization problem: maximizing $\phi$ over the polytope $\mathcal{S}$.
Such problems are efficiently solvable using interior-point methods (the maximizer of $\phi$ is known as the analytic center of $\mathcal{S}$, see \cite{renegar} and \cite{nn}) and the ellipsoid method (e.g., see \cite{yudinellipsoid, nemirovskiellipsoid, shorellipsoid}). See \cite{singh2014entropy} and \cite{leake2020computability} for further discussion on the computational complexity of solving such max-entropy optimization problems.
See also \cite{maxdet} for optimization techniques applied to the specific problem of determinant maximization.
Finally, see \cite{bhgaussian} for further discussion in the polytope case.

\begin{remark}
    The function $\phi(P)$ used above is the entropy of the associated Wishart distribution with expectation $P$.
    That is, it is \emph{not} the von Neumann entropy of the matrix $P$, which a priori might seem to be the natural measure of entropy for a given PSD matrix.
    Using the Wishart entropy instead of the von Neumann entropy is crucial to our analysis, and a similar situation can be found in \cite{leake2020computability}.
    See \cite{leake2020computability} for further discussion on why the von Neumann entropy is the incorrect measure of entropy in this case.
\end{remark}

\subsection{Step 2: The random variable $Y = AX$}

Recall the spectrahedron $\mathcal{S}$ is defined as the affine slice of $\PSD(N)$ given by $AP = b$, where $AP = (\tr(A_1P),\ldots,\tr(A_mP))$.
%
Let $X$ be a random variable with density function $f_X$ distributed according to the max-entropy distribution on $\PSD(N)$ discussed above, with expectation contained in $\mathcal{S}$.
We now construct a new random variable $Y = AX$ on $\R^m$ with density function $f_Y$.
%
Since $Y$ is defined as a linear map applied to $X$, we can use the standard change of variables formula for evaluating the density function of $Y$ at $b$:
\vspace{-0.5em}
\[
    f_Y(b) = \int_{A^{-1}(b)} f_X(P) \frac{\diff P}{\sqrt{\det(AA^\top)}} = \int_{\mathcal{S}} f_X(P) \frac{\diff P}{\sqrt{\det(AA^\top)}} = \vol(\mathcal{S}) \frac{f_X(P_0)}{\sqrt{\det(AA^\top)}},
    \vspace{-0.5em}
\]
where $P_0$ is any point of $\mathcal{S}$.

After computing $P^\star$ as above, we have $f_X(P) = \frac{1}{Z} e^{-\tr((P^\star)^{-1} P)}$ where the appropriate normalizing constant $Z$ is somewhat complicated, but can be computed explicitly by considering the expression for the density function of the Wishart distribution.
%
A straightforward computation then yields $f_X(P^\star) = e^{-\phi(P^\star)}$, where $\phi(P^\star)$ is defined above as the optimal entropy value.
We then rearrange the above expression to achieve the following formula for the volume of $\mathcal{S}$:
\vspace{-0.5em}
\[
    \vol(\mathcal{S}) = f_Y(b) \cdot e^{\phi(P^\star)} \sqrt{\det(AA^\top)}.
    \vspace{-0.5em}
\]
What remains then is to compute $f_Y(b)$.

\subsection{Step 3: Approximating the density function of $Y$} \label{sec:step_3}

To achieve an approximation formula for the volume of $\mathcal{S}$, the above expression shows we now only need to approximate $f_Y(b)$, the density function of $Y = AX$ at its expectation $b$.
For the polytope case, this is the part of the argument that requires the most technical work (see Section 6 in \cite{bhgaussian}), and a substantial part of this paper is then dedicated to generalizing their arguments to spectrahedra (see Section \ref{sec:proof_main}).

The overarching idea is that, since $X$ is distributed according to a max-entropy distribution on $\PSD(N)$, it is reasonable to guess that $Y$ is distributed according to a max-entropy distribution on $\R^m$.
Max-entropy distributions on $\R^m$ are precisely \emph{multivariate Gaussian distributions}, and thus we assume that $Y$ is distributed according to a multivariate Gaussian with expectation $b$ and covariance matrix $\Omega_Y$ of $Y$.
We can then hope to approximate $f_Y(b)$ by simply evaluating the density function of this Gaussian at its expectation.
To do this, we first compute the covariance matrix of $Y$.
Note first that since $X$ is distributed according to the Wishart distribution with expectation $P^\star$ (and $N+1$ degrees of freedom), we may write
\vspace{-0.5em}
\[
    X = \frac{1}{N+1} \sqrt{P^\star} GG^\top \sqrt{P^\star},
    \vspace{-0.5em}
\]
where $G$ is an $N \times (N+1)$ random matrix with independent standard Gaussian entries.
If $u_{ij}$ is the $(i,j)^\text{th}$ coordinate of $GG^\top$, then
\[
    \cov(u_{ij}, u_{i'j'}) = (N+1) \cdot \left(\delta_{i=i',j=j'} + \delta_{i=j',j=i'}\right)
\]
by bilinearity of $\cov$ and standard computations with independent normal random variables.
Recalling $Y = AX = \left(\tr(A_1X), \ldots, \tr(A_mX)\right)$ and defining $Z_k := \sqrt{P^\star} A_k \sqrt{P^\star}$, a straightforward computation using the above covariance computation for $u_{ij}$ yields
\[
    \cov(Y_i,Y_j) = \frac{1}{(N+1)^2} \cov\left(\tr(Z_i GG^\top), \tr(Z_j GG^\top)\right) = \frac{2}{N+1} \tr\left(Z_i Z_j\right).
\]
Thus the covariance matrix of $Y$ is $\frac{2}{N+1}$ times the Gram matrix of the matrices $Z_1,\ldots,Z_m$.
Letting $B$ be the linear map from $N \times N$ matrices to $\R^m$ defined by $BP := (\tr(Z_1P),\ldots,\tr(Z_mP))$,
this implies $\Omega_Y = \frac{2}{N+1} BB^\top$.
Computing the density of the appropriate Gaussian evaluated at its expectation then gives
\vspace{-0.75em}
\[
    f_Y(b) \approx \det(2\pi \cdot \Omega_Y)^{-1/2} = \left(\frac{N+1}{4\pi}\right)^{m/2} \left(\frac{1}{\det(BB^\top)}\right)^{1/2}.
    \vspace{-0.25em}
\]
This yields the following approximation formula for the volume of $\mathcal{S}$:
\[
    \vol(\mathcal{S}) \approx \left(\frac{N+1}{4\pi}\right)^{m/2} \left(\frac{\det(AA^\top)}{\det(BB^\top)}\right)^{1/2} e^{\phi(P^\star)},
\]
where $\phi(P^\star)$ is the maximum entropy value, as discussed above. This is precisely the formula given in Section \ref{sec:main_result}.


What remains then is to justify our assumption that $Y$ is distributed like a multivariate Gaussian with expectation $b$ and covariance matrix $\Omega_Y$.
We do this by proving a local central limit theorem for $Y$ in a number of situations.
In the polytope case, the argument in \cite{bhgaussian} relies heavily on the \emph{well-roundedness} of $B$, which is described by the ratio of the maximum column norm of $B$ to the minimum singular value of $B$.
This quantity resembles the condition number of $\sqrt{BB^\top} = \sqrt{\Omega_Y}$, but is crucially different, as replacing the well-roundedness by the condition number of $\sqrt{\Omega_Y}$ yields vacuous results in \cite{bhgaussian}.

To prove our main results for spectrahedra, we first prove a technical result (Theorem \ref{thm:main}) which is analogous to the main result of \cite{bhgaussian}.
This result relies upon the use of a similar well-roundedness quantity, suited to spectrahedra.
One problem with this is that the maximum column norm of $B$ in this case has little to no conceptual meaning.
%
Thus we replace the maximum column norm by
\vspace{-0.5em}
\[
    \max_{\|x\|_2 = 1} \|(x^\top Z_k x)_{k=1}^m\|_2,
\]
where the matrices $Z_k$ are as defined above.
Note that this corresponds the maximum column norm of $B$ whenever the matrices $Z_k$ are diagonal, and thus this leads to a natural matrix-theoretic generalization of the well-roundedness. 
%
%
That said, this quantity remains a bit of a mystery to us, and a more conceptual understanding of it would likely lead to improved results or simpler proofs.

A good upper bound on the well-roundedness of $B$ then enables a proof of the local central limit theorem, which implies the above approximate volume formula whenever the dimension of $\mathcal{S}$ is large.
(Conceptually, the largeness of the dimension is required because the local central limit theorem only applies in the limit.)
This is the last piece of the argument which leads to the main technical result (Theorem \ref{thm:main}), from which we derive our main results.

%
%

\subsection{From the main technical result to the main results} \label{sec:technical-to-main}

We now discuss how our main technical result (Theorem \ref{thm:main}) discussed above implies our main results in Section \ref{sec:main_result}.
The key idea is to remove reliance upon the well-roundedness quantity discussed in the previous section, by replacing it with something conceptually simpler.
It is straightforward from the definition that well-roundedness of $B$ can be upper bounded by the condition number of the matrix $\sqrt{BB^\top}$.
Thus in theory, we can replace the well-roundedness of $B$ by this condition number, at the expense that this will weaken the result.
And even better: by a straightforward change of variables, we can assume further that the conjugated constraint matrices $Z_k = \sqrt{P^\star} A_k \sqrt{P^\star}$ are orthonormal, without changing the spectrahedron or the approximation formula (see Lemma \ref{lem:spectra_repr}).
This implies $BB^\top = I_m$, and thus we can replace the condition number of $\sqrt{BB^\top}$ with its optimal value of 1.

The crucial difference between the polytope and spectrahedron case is then what happens when we make this replacement.
%
As discussed above, this replacement leads to vacuous results in the polytope case.
This is due to the fact that applying Theorem~2.2 of \cite{bhgaussian} roughly requires
\vspace{-0.5em}
\[
    \omega^2 \cdot m^2 \log n \to 0
    \vspace{-0.5em}
\]
asymptotically, where $n$ is the dimension of the positive orthant where the polytope $\mathcal{P}$ lies, $m$ is the number of affine constraints which cut out $\mathcal{P}$, and $\omega$ is the well-roundedness of $B$. We can replace $\omega$ with the optimal condition number value of 1,
but the problem is then that $m^2 \log n\to 0$ is never satisfied.

On the other hand, the requirement in the spectrahedron case (see Condition \ref{eq:gamma_assumption} of Theorem \ref{thm:main}) essentially becomes
\vspace{-0.5em}
\[
    \omega^2 \cdot \frac{m^3 \log N}{N} \to 0.
\]
We can then replace $\omega$ with the optimal condition number value of 1, and it is still a priori possible that this requirement is satisfied. In fact, it is satisfied for many interesting classes of spectrahedra (such as those having $m = O(1)$), and the main results of Section \ref{sec:main_result} spell out various general cases that follow as straightforward corollaries; see Section \ref{sec:proofs_corollaries}.

Why this difference between the polytope case and the spectrahedron case? The denominator factor of $N$ in the above requirement for spectrahedra is the key difference, and it is related to the extra factor of $\frac{2}{N+1}$ that appears in the expression for the covariance matrix of $Y$ in Section \ref{sec:step_3} above. (In the polytope case, where exponential distributions are considered, an analogous factor does not appear.) That said, it is not immediately clear to us philosophically why this difference occurs. It appears in the computations of the proof of Theorem \ref{thm:main} due to explicit differences in the characteristic functions of the Wishart distribution and of the exponential distribution. In particular, an extra factor of $\frac{2}{N+1}$ appears for the Wishart distrbution; see Section \ref{sec:char_small_t}. This suggests to us that similar results to those presented in \cite{bhgaussian} and in this paper may be possible for affine slices of symmetric cones in general (e.g., cones of Hermitian PSD matrices and cones of quaternionic PSD matrices). The rank of a given symmetric cone may then have some connection to these extra factors that appear. We do not discuss this any further here.

\begin{remark}
    The difference between the polytope and spectrahedron cases discussed above may seem a bit strange since spectrahedra are generalizations of polytopes.
    That said, our results for spectrahedra do not actually improve the results in the polytope case.
    Representing a polytope via a spectrahedral representation requires affine constraints which zero out all off-diagonal entries; this would force $m$ to be of the order $N^2$ and cause the above required condition to become vacuous (as in the polytope case).
\end{remark}

\section{Maximum Entropy Distributions over the PSD Cone} \label{sec:max-entropy}

In this section, we construct a probability distribution which maximizes differential entropy over probability distributions on $\PSD(N)$ with a given fixed expectation.
This probability distribution is the \emph{Wishart distribution} $W_N\left(\frac{P}{N+1}, N+1\right)$ on $\PSD(N)$ for any choice of expectation $P \in \PD(N)$ (see Theorem \ref{thm:strong_duality}).
This is in fact precisely the reason why the function $\phi$ appears in Theorem \ref{thm:main} above; specifically, the density function of the Wishart distribution with expectation $P$ evaluated at $P$ is given by $e^{-\phi(P)}$ (see Theorem \ref{thm:strong_duality} and Definition \ref{def:entropy_function}).

The way we utilize this fact to approximate the volume of spectrahedra then goes as follows.
Suppose $P^\star$ maximizes $\phi(P)$ over the spectrahedron $\mathcal{S}$.
Then in fact the Wishart distribution with expectation $P^\star$ given by $W_N\left(\frac{P^\star}{N+1}, N+1\right)$ has constant density on all of $\mathcal{S}$ (see Corollary \ref{cor:density_formula}).
In particular, this means that the volume of $\mathcal{S}$ is inversely proportional to the density function of this Wishart distribution evaluated at $P^\star \in \mathcal{S}$, which is given by $e^{-\phi(P^\star)}$ as mentioned above.
All that remains is then to determine the constant of proportionality, and this is discussed in detail in Section~\ref{sec:proof_main}.

The remainder of this section is spent proving the facts we describe above.
We start with a few technical lemmas, and then prove the main results of this section in Theorem \ref{thm:strong_duality} and Corollary \ref{cor:density_formula}.

\begin{lemma}[see \cite{guler1996barrier}] \label{lem:Guler}
    Given $Y \in \PD(N)$, we have
    \[
        \log \int_{\PSD(N)} e^{-\tr(YQ)} \, \diff Q = -\frac{N+1}{2} \log \det(Y) + C
    \]
    with $C = \log \int_{\PSD(N)} e^{-\tr(Q)} \, \diff Q = \log \Gamma_N(\frac{N+1}{2})$, where $\Gamma_N$ is the multivariate gamma function.
\end{lemma}
\begin{proof}
    We may write $\tr(YQ)= \tr(\sqrt{Y} Q \sqrt{Y}).$ Note that $Q\mapsto \sqrt{Y}Q\sqrt{Y}$ is an endomorphism of $\Sym(N)$ which preserves $\PSD(N)$. If $\lambda_1,\lambda_2,\dots,\lambda_N$ denote the eigenvalues of $Y$, then the Jacobian of the map $Q\mapsto \sqrt{Y}Q\sqrt{Y}$ equals \[
    \Big(\prod_{i=1}^N \lambda_i\Big)\, \Big(\prod_{i\neq j} \sqrt{\lambda_i}\sqrt{\lambda_j}\Big) = \det(Y)^{(N+1)/2}
    \] which can be seen by diagonalizing $Y$. Hence, applying the change of variables $P=\sqrt{Y}Q\sqrt{Y}$ to the above integral, the equality \[
    \int_{\PSD(N)} e^{-\tr(YQ)}\,\diff Q = \det(Y)^{-(N+1)/2}\,\int_{\PSD(N)} e^{-\tr(P)}\,\diff P
    \] holds. The result follows from taking logarithm of both sides.
\end{proof}

\begin{corollary} \label{cor:dual_unique_minimizer}
    Given $P \in \PD(N)$, the convex optimization problem
    \[
        \inf_{Y \in \PD(N)} \left[\tr(YP) + \log \int_{\PSD(N)} e^{-\tr(YQ)} \, dQ\right]
    \]
    attains its unique minimum at $Y^\star = \frac{N+1}{2} P^{-1}$.
\end{corollary}
\begin{proof}
    By Lemma \ref{lem:Guler}, we can write an equivalent optimization problem as
    \[
        \inf_{Y \in \PD(N)} \left[\tr(YP) - \frac{N+1}{2} \log \det(Y)\right].
    \]
    This objective function is continuous in $Y \in \PD(N)$, and by a standard computation we now compute its gradient via
    \[
        \nabla\left[\tr(YP) - \frac{N+1}{2} \log \det(Y)\right] = P - \frac{N+1}{2} Y^{-1}.
    \]
    Therefore the gradient is 0 if and only if $Y = \frac{N+1}{2} P^{-1}$, and thus the objective function attains its unique minimum at $Y^\star = \frac{N+1}{2} P^{-1}$.
\end{proof}

The following is proven in more generality in \cite{leake2020computability}, but we give a direct proof for the specific case of distributions on $\PSD(N)$ for the sake of the reader. In particular, we do not actually prove here that the following optimization problems are a primal-dual pair with strong duality. We instead simply demonstrate that their optimal values are equal.

\begin{theorem} \label{thm:strong_duality}
    For $P \in \PD(N)$, let us define $\phi(P)$ to be the maximum differential entropy of a probability distribution on $\PSD(N)$ with expectation $P$; that is,
    \[
    \begin{split}
        \phi(P) = \sup &\int_{\PSD(N)} -f(Q) \log f(Q) \, \diff Q \\
        \text{subject to:} \quad &f(P) > 0 \quad \text{for all} \quad P \in \PD(N) \\
            &\int_{\PSD(N)} f(Q) \, \diff Q = 1 \\
            &\int_{\PSD(N)} f(Q) \, Q \, \diff Q = P.
    \end{split}
    \]
    Then $\phi(P)$ is equal to the dual optimum of the above primal optimization problem; that is,
    \[
        \phi(P) = \inf_{Y \in \PD(N)} \left[\tr(YP) + \log \int_{\PSD(N)} e^{-\tr(YQ)} \, \diff Q\right].
    \]
    For fixed $P \in \PD(N)$, there is a probability distribution $\mu^\star$ on $\PSD(N)$ with differential entropy $\phi(P)$ and expectation $P$, and with density function $f^\star$ which can be expressed via
    \[
        f^\star(Q) = \frac{e^{-\tr(Y^\star Q)}}{\int_{\PSD(N)} e^{-\tr(Y^\star R)} \, \diff R},
    \]
    where $Y^\star \in \PD(N)$ is the unique minimizer $Y^\star$ of the dual formulation.
\end{theorem}
\begin{proof}
    By Corollary \ref{cor:dual_unique_minimizer}, the dual formulation of $\phi(P)$ attains its unique minimum at $Y^\star = \frac{N+1}{2} P^{-1}$.
    We first show that $f^\star(Q)$ is a maximizer for the primal (maximum entropy) optimization problem.
    Since differential entropy is concave, we compute the derivative of the primal objective function in all directions orthogonal to the two linear constraints of primal optimization problem.
    Specifically, let $g$ be bounded and compactly supported on $\PSD(N)$, and such that
    \[
        \int_{\PSD(N)} g(Q) \, \diff Q = 0 \qquad \text{and} \qquad \int_{\PSD(N)} Q \cdot g(Q) \, \diff Q = 0.
    \]
    Letting $f^\star$ be defined as above, we now compute
    \[
    \begin{split}
        \left.\frac{\diff}{\diff t}\right|_{t=0} -\int_{\PSD(N)} &(f^\star+t \cdot g)(Q) \log (f^\star+t \cdot g)(Q) \, \diff Q \\
            &= -\int_{\PSD(N)} g(Q) \log f^\star(Q) \, \diff Q - \int_{\PSD(N)} f^\star(Q) \cdot \frac{g(Q)}{f^\star(Q)} \, \diff Q \\
            &= -\int_{\PSD(N)} g(Q) \log\left[\frac{e^{-\tr(Y^\star Q)}}{\int_{\PSD(N)} e^{-\tr(Y^\star R)} \, \diff R}\right] \, \diff Q \\
            &= \int_{\PSD(N)} g(Q) \cdot \tr(Y^\star Q) \, \diff Q \\
            &= \tr\left(Y^* \int_{\PSD(N)} Q \cdot g(Q) \, \diff Q\right) \\
            &= 0.
    \end{split}
    \]
    To show that $f^\star$ is a maximizer for the primal problem, we just need to show that it satisfies all the necessary constraints.
    First, that $f^\star(Q) > 0$ for all $Q \in \PD(N)$ and that $\int_{\PSD(N)} f^\star(Q) \, \diff Q = 1$ are immediate from the definition of $f^\star$.
    To show that $\int_{\PSD(N)} Q \cdot f^\star(Q) \, \diff Q = P$, we compute the gradient of the dual program directly via
    \[
    \begin{split}
        0 = \left.\nabla\right|_{Y=Y^\star} \left[\tr(YP) + \log \int_{\PSD(N)} e^{-\tr(YQ)} \, \diff Q\right] &= P - \frac{\int_{\PSD(N)} Q \cdot e^{-\tr(Y^\star Q)} \, \diff Q}{\int_{\PSD(N)} e^{-\tr(Y^\star Q)} \, \diff Q} \\
            &= P - \int_{\PSD(N)} Q \cdot \frac{e^{-\tr(Y^\star Q)}}{\int_{\PSD(N)} e^{-\tr(Y^\star R)} \, \diff R} \, \diff Q \\
            &= P - \int_{\PSD(N)} Q \cdot f^\star(Q) \, \diff Q.
    \end{split}
    \]
    That is, $f^\star$ is a maximizer of the primal (maximum entropy) optimization problem.
    
    
    We now show that the optimal values of the primal and dual problems are equal.
    Since we have determined optimal inputs for both problems, we simply compute the objective functions and show that they are equal.
    We compute
    \[
    \begin{split}
        -\int_{\PSD(N)} &f^\star(Q) \log f^\star(Q) \, dQ \\
            &= -\int_{\PSD(N)} \frac{e^{-\tr(Y^\star Q)}}{\int_{\PSD(N)} e^{-\tr(Y^\star R)} \, dR} \log\frac{e^{-\tr(Y^\star Q)}}{\int_{\PSD(N)} e^{-\tr(Y^\star R)} \, dR} \, dQ \\
            &= \int_{\PSD(N)} \frac{e^{-\tr(Y^\star Q)}}{\int_{\PSD(N)} e^{-\tr(Y^\star R)} \, dR} \cdot \tr(Y^\star Q) \, dQ + \log \int_{\PSD(N)} e^{-\tr(Y^\star Q)} \, dQ \\
            &= \tr\left(Y^\star \int_{\PSD(N)} Q \cdot \frac{e^{-\tr(Y^\star Q)}}{\int_{\PSD(N)} e^{-\tr(Y^\star R)} \, dR} \, dQ\right) + \log \int_{\PSD(N)} e^{-\tr(Y^\star Q)} \, dQ \\
            &= \tr(Y^\star P) + \log \int_{\PSD(N)} e^{-\tr(Y^\star Q)} \, dQ.
    \end{split}
    \]
    This is precisely the dual objective function evaluated at $Y = Y^\star$, and this completes the proof.
\end{proof}

The function $\phi$ evaluated at its maximizer $P^\star$ gives the maximum entropy value, and we have shown in Theorem \ref{thm:strong_duality} above that this is also the optimal value of the dual convex program. We now formally name this function $\phi$, and then show that it is equal to the expression for $\phi$ given in Section \ref{sec:main_result}.

\begin{definition} \label{def:entropy_function}
    We refer to the function $\phi$ from Theorem \ref{thm:strong_duality} as the \emph{maximum entropy function} on $\PD(N)$.
\end{definition}

\begin{corollary} \label{cor:entropy_function_formula}
    The maximum entropy function $\phi$ can be expressed as
    \[
        \phi(P) = \log\Gamma_N\left(\frac{N+1}{2}\right) - \frac{N(N+1)}{2} \log\left(\frac{N+1}{2e}\right) + \frac{N+1}{2}\log\det(P).
    \]
\end{corollary}
\begin{proof}
    By Theorem \ref{thm:strong_duality}, Lemma \ref{lem:Guler}, and Corollary \ref{cor:dual_unique_minimizer}, we have
    \[
    \begin{split}
        \phi(P) &= \log\Gamma_N\left(\frac{N+1}{2}\right) + \tr(Y^\star P) - \frac{N+1}{2} \log\det(Y^\star) \\
            &= \log\Gamma_N\left(\frac{N+1}{2}\right) - \frac{N(N+1)}{2} \log\left(\frac{N+1}{2e}\right) + \frac{N+1}{2}\log\det(P),
    \end{split}
    \]
    where $Y^\star = \frac{N+1}{2} P^{-1}$.
\end{proof}

\begin{corollary} \label{cor:density_formula}
    Given real symmetric matrices $A_1,\ldots,A_m$ and real vector $\bm{b} \in \R^m$, define the corresponding spectrahedron
    \[
        \mathcal{S} = \{X \in \PSD(N)\,:\,\tr(A_i X) = b_i,~i \in [m]\}.
    \]
    Suppose that $\mathcal{S}$ is bounded.
    The maximum entropy function $\phi$ is concave and attains its unique maximum over $\mathcal{S}$ at a unique point $P^\star$ in the relative interior of $\mathcal{S}$.
    Further, let $\mu^\star$ (with density function $f^\star$) be the probability distribution on $\PSD(N)$ guaranteed by Theorem \ref{thm:strong_duality}, with differential entropy $\phi(P^\star)$ and expectation $P^\star$.
    Then $f^\star$ is constant on $\mathcal{S}$; specifically,
    \[
        f^\star(P) = e^{-\phi(P^\star)} \quad \text{for all} \quad P \in \mathcal{S}.
    \]
\end{corollary}
\begin{proof}
    Using the formula for the maximum entropy function $\phi$ from Corollary \ref{cor:entropy_function_formula}, we compute its gradient via
    \[
        \nabla \phi(P) = \nabla\left[\frac{N+1}{2}\log\det(P)\right] = \frac{N+1}{2} P^{-1}.
    \]
    Note first that for $P$ near the boundary of $\mathcal{S}$, the determinant of $P$ approaches 0 and thus $\phi(P)$ approaches $-\infty$.
    Therefore by boundedness of $\mathcal{S}$, it must be that $\phi$ is maximized at some point in the relative interior of $\mathcal{S}$.
    Concavity of $\phi$ and uniqueness of the maximizer then follow from the fact that $\log\det(P)$ is a strictly concave function on $\PD(N)$; see Example 11.7 of \cite{boyd2004convex}.
    Let us denote this unique maximizer by $P^\star$.
    
    Since $\phi$ is smooth on $\PD(N)$, the gradient of $\phi$ at $P^\star$ is 0 when restricted to $\mathcal{S}$.
    Equivalently, the gradient of $\phi$ at $P^\star$ must be orthogonal to the affine span of $\mathcal{S}$.
    That is, for all $P \in \mathcal{S}$ we have
    \[
    \begin{split}
        0 &= \tr\left(\nabla \phi(P^\star) \cdot (P-P^\star)\right) = \frac{N+1}{2} \tr\left((P^\star)^{-1} (P-P^\star)\right) \\
            &\implies \frac{N+1}{2}\tr\left((P^\star)^{-1} P\right) = \frac{N(N+1)}{2}.
    \end{split}
    \]
    By Corollary \ref{cor:dual_unique_minimizer} and Theorem \ref{thm:strong_duality}, we have that
    \[
        f^\star(Q) \propto e^{-\tr(Y^\star Q)} = e^{-\frac{N+1}{2}\tr((P^\star)^{-1} Q)},
    \]
    where $Y^\star$ is the unique minimizer of the dual formulation for $P^\star$.
    It immediately follows that $f^\star$ is constant on $\mathcal{S}$.
    To see that this constant is equal to $e^{-\phi(P^\star)}$, we use Corollary \ref{cor:entropy_function_formula} to compute
    \[
        e^{-\phi(P^\star)} = \frac{1}{\Gamma_N\left(\frac{N+1}{2}\right)} \left(\frac{N+1}{2e}\right)^{\frac{N(N+1)}{2}} \det(P^\star)^{-\frac{N+1}{2}}.
    \]
    By the above computations and Lemma \ref{lem:Guler}, we then have
    \[
    \begin{split}
        f^\star(P^\star) &= \frac{e^{-\frac{N+1}{2}\tr((P^\star)^{-1} P^\star)}}{\int_{\PSD(N)} e^{-\frac{N+1}{2}\tr((P^\star)^{-1} Q)} dQ} \\
            &= e^{-\frac{N(N+1)}{2}} \cdot e^{\frac{N+1}{2}\log\det\left(\frac{N+1}{2}(P^\star)^{-1}\right) - \log \Gamma_N\left(\frac{N+1}{2}\right)} \\
            &= \frac{1}{\Gamma_N\left(\frac{N+1}{2}\right)} \left(\frac{1}{e}\right)^{\frac{N(N+1)}{2}} \det\left(\frac{N+1}{2}(P^\star)^{-1}\right)^{\frac{N+1}{2}} \\
            &= \frac{1}{\Gamma_N\left(\frac{N+1}{2}\right)} \left(\frac{N+1}{2e}\right)^{\frac{N(N+1)}{2}} \det(P^\star)^{-\frac{N+1}{2}}.
    \end{split}
    \]
    That is, $f^\star(P^\star) = e^{-\phi(P^\star)}$ and thus $f^\star(P) = e^{-\phi(P^\star)}$ for $P \in \mathcal{S}$ since $f^\star$ is constant on $\mathcal{S}$.

\end{proof}

Finally, we note that the probability distribution $\mu^\star$ of Corollary \ref{cor:density_formula} actually maximizes differential entropy over all distributions with expectation contained in $\mathcal{S}$.
First, by Corollary \ref{cor:density_formula} we have that $P^\star$ uniquely maximizes $\phi(P)$ over all $P \in \mathcal{S}$.
Then, since $\phi$ can be represented as the optimal value of the maximum entropy optimization problem given in Theorem \ref{thm:strong_duality}, we have that any probability distribution on $\PSD(N)$ with differential entropy $\phi(P^\star)$ and expectation contained in $\mathcal{S}$ must actually have expectation equal to $P^\star$ (or else $P^\star$ would not uniquely maximize $\phi$).

\section{Basic Examples}

\subsection{The spectraplex}

We first consider a simple example for which there is an explicit volume formula. Fixing $A \in \PD(N)$, we consider
\[
    \mathcal{S} := \left\{P \in \PSD(N) ~:~ \tr(AP) = 1\right\}.
\]
In this case that $A$ is the identity matrix, this spectrahedron is called the \emph{spectraplex} or the set of density matrices. More generally, the spectrahedron $\mathcal{S}$ generalizes the simplex in the sense that any simplex (up to translation) can be given as the intersection of a single affine hyperplane with the positive orthant.

The spectrahedron $\mathcal{S}$ has an explicit volume formula, which we compute as follows. First we write
\[
    \Gamma_N\left(\frac{N+1}{2}\right) = \int_{\PD(N)} e^{-\tr(X)} \diff X = \det(A)^{\frac{N+1}{2}} \int_{\PD(N)} e^{-\tr(AX)} \diff X.
\]
Next note that the line $f(t) = \frac{A}{\|A\|_F} t$ is the unique unit-speed line through 0 which is orthogonal to $\mathcal{S}$, and $f(\|A\|_F^{-1}) \in \mathcal{S}$. Thus we have
\[
    \int_{\PD(N)} e^{-\tr(AX)} \diff X = \vol(\|A\|_F \cdot \mathcal{S}) \int_0^\infty e^{-t\|A\|_F} t^{\dim(\mathcal{S})} \,\diff t = \dim(\mathcal{S})! \cdot \|A\|_F^{-1} \vol(\mathcal{S}).
\]
Combining and rearranging gives
\[
    \vol(\mathcal{S}) = \Gamma_N\left(\frac{N+1}{2}\right) \frac{\|A\|_F}{\det(A)^{\frac{N+1}{2}}} \left(\binom{N+1}{2} - 1\right)!^{-1}.
\]
Using Stirling's approximation and the fact that $(k-1)! = \frac{k!}{k}$, we can obtain the following asymptotic formula for large $N$:
\[
    \vol(\mathcal{S}) \approx \Gamma_N\left(\frac{N+1}{2}\right) \frac{\|A\|_F}{\det(A)^{\frac{N+1}{2}}} \left(\frac{2e}{N(N+1)}\right)^{\binom{N+1}{2}} \left(\frac{N(N+1)}{4\pi}\right)^{\frac{1}{2}}.
\]
We now apply the asymptotic formula of Theorem \ref{thm:asymptotic} and compare the result to the above formula.

\paragraph{Entropy function maximizer.} We claim in this case that $P^\star = (NA)^{-1}$. To see this, recall that the gradient of $\log \det P$ is $P^{-1}$ since we are restricting to symmetric matrices. For any $P \in \mathcal{S}$, we then have
\[
    \langle P-P^\star, \nabla \phi(P^\star) \rangle = \tr(P(P^\star)^{-1}) - N = \tr(NAP) - N = 0.
\]
It is then clear that $P^\star \in \mathcal{S}$.

\paragraph{The asymptotic volume.} Since $\mathcal{S}$ is defined by a constant number of constraints, we apply Theorem \ref{thm:asymptotic} to obtain an asymptotic formula for the volume of $\mathcal{S}$. For large $N$, Theorem \ref{thm:asymptotic} gives
\[
\begin{split}
    \vol(\mathcal{S}) &\approx \left(\frac{N+1}{4\pi}\right)^{\frac{1}{2}} \left(\frac{\tr(A^2)}{N^{-1}}\right)^{\frac{1}{2}} e^{\phi(P^\star)} \\
        &= \left(\frac{N(N+1)}{4\pi}\right)^{\frac{1}{2}} \|A\|_F \cdot \Gamma_N\left(\frac{N+1}{2}\right) \cdot \left(\frac{2e}{N+1}\right)^{\binom{N+1}{2}} \det(NA)^{-\frac{N+1}{2}} \\
        &= \Gamma_N\left(\frac{N+1}{2}\right) \frac{\|A\|_F}{\det(A)^{\frac{N+1}{2}}} \left(\frac{2e}{N(N+1)}\right)^{\binom{N+1}{2}} \left(\frac{N(N+1)}{4\pi}\right)^{\frac{1}{2}}.
\end{split}
\]
This is the same expression that we obtained above using an explicit formula for the volume, and thus this proves our asymptotic result in this special case.

\subsection{One PD constraint and one rank-one constraint}

Let us now consider another simple example. Fix $A \in \PD(N)$ and $v \in \R^N$ such that $\xi := v^\top A^{-1} v > 1$. We consider the case given by
\[
    \mathcal{S} := \left\{P \in \PSD(N) ~:~ \tr(AP) = 1 \text{ and } \tr(vv^\top P) = v^\top P v = 1\right\}.
\]

\paragraph{Entropy function maximizer.} We claim in this case that $P^\star = (aA + bvv^\top)^{-1}$, where
\[
    a = \frac{(N-1) \xi}{\xi - 1} \qquad \text{and} \qquad b = \frac{\xi - N}{\xi - 1}.
\]
To see this, recall that the gradient of $\log\det P$ is $P^{-1}$ since we are restricting to real symmetric matrices. For any $P \in \mathcal{S}$, we then have
\[
    \langle P-P^\star, \nabla \phi(P^\star) \rangle = \tr(P(P^\star)^{-1}) - N = \tr(aAP) + \tr(bvv^\top P) - N = a + b - N = 0.
\]
Thus we only have left to show that $P^\star \in \mathcal{S}$. To see that $P^\star \in \PD(N)$, we first write
\[
    (P^\star)^{-1} = \left.\frac{(N-1)\xi A + t vv^\top}{\xi - 1}\right|_{t = \xi - N}.
\]
If $\xi - N > 0$, then $(P^\star)^{-1}$ is positive definite and thus so is $P^\star$. Otherwise recall the following consequence of the matrix determinant lemma for invertible $B$:
\[
    B + uu^\top \text{ is invertible} \quad \iff \quad u^\top B^{-1} u \neq -1.
\]
Applying this to the above expression, we have
\[
    M(t) := (N-1)\xi A + t vv^\top \text{ is invertible} \quad \iff \quad \frac{t \cdot v^\top A^{-1} v}{(N-1)\xi} \neq -1 \quad \iff \quad t \neq 1 - N.
\]
Since $\xi - N > 1 - N$, we have that $M(t)$ is invertible for all $t \in [\xi - N, 0]$. Since $M(0)$ is positive definite, we therefore have that $M(\xi-N)$ is also positive definite and thus so is $P^\star$. We now finally apply the Sherman-Morrison formula to get
\[
    \tr(AP^\star) = \tr\left(A \left[a^{-1}A^{-1} - \frac{a^{-2}bA^{-1}vv^\top A^{-1}}{1 + a^{-1}bv^\top A^{-1}v}\right]\right) = \frac{N}{a} - \frac{b\xi}{a^2+ab\xi} = 1
\]
and
\[
    \tr(vv^\top P^\star) = \tr\left(vv^\top \left[a^{-1}A^{-1} - \frac{a^{-2}bA^{-1}vv^\top A^{-1}}{1 + a^{-1}bv^\top A^{-1}v}\right]\right) = \frac{\xi}{a} - \frac{b\xi^2}{a^2+ab\xi} = 1.
\]
Therefore $P^\star$ is the optimizer of $\phi$ over $\mathcal{S}$.

\paragraph{Applying Theorem \ref{thm:approx}.} Since $\mathcal{S}$ is defined by a constant number of constraints, we apply Theorem \ref{thm:approx} which says that we can apply the volume formula with relative error $\epsilon$ whenever $0 < \epsilon \leq \frac{1}{2}$ and
\[
    \epsilon^2 \geq \frac{8 \cdot 10^5 \left(2 + \log(\epsilon^{-1})\right)^2 \log(N\epsilon^{-1})}{N+1}.
\]
Setting $\epsilon := \frac{10^3 \log^{3/2}(N+1)}{(N+1)^{1/2}}$, we have
\[
\begin{split}
    \left(2 + \log(\epsilon^{-1})\right)^2 \log(N\epsilon^{-1}) &\leq \log^2\left(\frac{e^2 (N+1)^{1/2}}{10^3 \log^3(N+1)}\right) \log\left(\frac{(N+1)^{3/2}}{10^3 \log^3(N+1)}\right) \\
        &\leq \frac{3}{8} \log^3(N+1),
\end{split}
\]
which implies
\[
    \epsilon^2 = \frac{10^6 \log^3(N+1)}{N+1} \geq \frac{8 \cdot 10^5}{N+1} \left(\frac{3}{8} \log^3(N+1)\right) \geq \frac{8 \cdot 10^5 \left(2 + \log(\epsilon^{-1})\right)^2 \log(N\epsilon^{-1})}{N+1}.
\]
Thus we can apply Theorem \ref{thm:approx} whenever $\epsilon \leq \frac{1}{2}$, which is satisfied whenever $N \geq 10^{11}$, for example.

\paragraph{The approximate volume.}  We now compute the approximate volume of $\mathcal{S}$ for large enough $N$. Since $m = 2$, we have
\[
    \vol(\mathcal{S}) \approx \left(\frac{N+1}{4\pi}\right) \left(\frac{\det(AA^\top)}{\det(BB^\top)}\right)^{1/2} e^{\phi(P^\star)},
\]
where
\[
    \det(AA^\top) = \|A\|_F^2 \cdot \|v\|_2^4 - (v^\top A v)^2
\]
and, after a straightforward computation using the Sherman-Morrison formula,
\[
    \det(BB^\top) = \tr\left((AP^\star)^2\right) \cdot \tr\left((vv^\top P^\star)^2\right) - \left(\tr(AP^\star vv^\top P^\star)\right)^2 = \frac{(N-1)\xi^2}{a^2(a+b\xi)^2} = \frac{(\xi-1)^2}{(N-1)\xi^2}.
\]
Thus, we have that
\[
    \left(\frac{\det(AA^\top)}{\det(BB^\top)}\right)^{1/2} = \frac{\xi (N-1)^{1/2}}{\xi - 1} \left(\|A\|_F^2 \cdot \|v\|_2^4 - (v^\top A v)^2\right)^{1/2}.
\]
Further, we have
\[
    e^{\phi(P^\star)} = \Gamma_N\left(\frac{N+1}{2}\right) \cdot \left(\frac{2e}{N+1}\right)^{\binom{N+1}{2}} \det(aA + bvv^\top)^{-\frac{N+1}{2}},
\]
where, by the matrix determinant lemma,
\[
    \det(aA + bvv^\top) = (1+a^{-1}b\xi) \cdot \det(aA) = \frac{\xi-1}{N-1} \cdot \left(\frac{(N-1)\xi}{\xi-1}\right)^N \det(A).
\]
Combining everything gives
\[
\begin{split}
    \vol(\mathcal{S}) &\approx \frac{e}{2\pi} (N-1)^{\frac{N}{2}} \left(\frac{2e}{N^2-1}\right)^{\binom{N+1}{2}-1} \Gamma_N\left(\frac{N+1}{2}\right) \\
        &\times \left(\frac{1}{\xi-1}\right)^{\frac{N+1}{2}} \left(\frac{\xi-1}{\xi}\right)^{\binom{N+1}{2} - 1} \left(\frac{\|A\|_F^2 \cdot \|v\|_2^4 - (v^\top A v)^2}{\det(A)^{N+1}}\right)^{1/2}.
\end{split}
\]

\subsection{Diagonal constraints}

We compute one quick final example before moving on to the main exampling given in Section \ref{sec:multi-CP-maps} below. Consider the problem of computing the volume of the spectrahedron given as follows
\[
\mathcal{S}_{\alpha, \beta} = \{P \in PSD(2N)\, : \, \operatorname{Tr} (M_1P) = \alpha, \, \operatorname{Tr}(M_2P) = \beta\},\]
where the matrices $M_i$ are given by 
\[
    M_1 = \begin{bmatrix} I_N & 0\\0 & 0\end{bmatrix},\quad M_2 = \begin{bmatrix} 0 & 0\\0 & I_N\end{bmatrix}.
\] Equivalently, \[
\mathcal{S}_{\alpha,\beta} = \Big\{ \begin{bmatrix}
    P_1 & P_2\\
    P_3 & P_4
\end{bmatrix}\in\PD(2N) \; :\; \tr(P_1)=\alpha ,\; \tr(P_4)=\beta\, \Big\}.
\]
A straightforward calculation shows that the entropy function maximizer $P^\star$ is given by 
\[
    P^\star = \dfrac{1}{N}\begin{bmatrix} \alpha I_N & 0\\0 & \beta I_N\end{bmatrix}.
\]
The linear operators $AA^{T}$ and $BB^{T}$ on $\mathbb{R}^2$ become respectively
\[
    AA^\top = N
    \begin{bmatrix}
        1 & 0\\
        0 & 1
    \end{bmatrix}
    , \quad BB^\top = \frac{1}{N}
    \begin{bmatrix}
        \alpha^2 & 0\\
        0 & \beta^2
    \end{bmatrix}.
\]
Since $m=2$, we can apply Theorem \ref{thm:asymptotic} to obtain the asymptotic expression
\[
    \vol(\mathcal{S}_{\alpha,\beta}) \approx \frac{N^2(2N+1)}{4\pi\alpha\beta} e^{\phi(P^\star)} = \frac{N^2(2N+1)}{4\pi\alpha\beta} \left(\frac{2e\sqrt{\alpha\beta}}{N(2N+1)}\right)^{N(2N+1)} \Gamma_{2N}\left(\frac{2N+1}{2}\right)
\]
for $N \to \infty$.


\section{Main Example: Multi-stochastic Completely Positive Maps} \label{sec:multi-CP-maps}

In this section, we apply our main asymptotic result (Theorem \ref{thm:asymptotic}) to a spectrahedral generalization of the multi-way Birkhoff polytope. This spectrahedron consists of real symmetric positive definite matrices which are naturally associated to completely positive linear maps on matrices with certain stochasticity properties. Such ``multi-stochastic'' completely positive maps generalize doubly stochastic completely positive maps, which in turn generalize doubly stochastic matrices (i.e., points of the Birkhoff polytope). As discussed in the introduction, this spectrahedron can also be described as the set of all quantum states with maximal entanglement; i.e., the set of all quantum states with all univariant quantum marginals equal to the identity matrix.





\subsection{Transportation polytopes and completely positive maps}

In \cite{bhgaussian}, the main application of their technical results is to asymptotic integer point counting and volume computation for transportation polytopes and their generalizations. Recall that a \emph{transportation polytope} is defined as the set of $m \times n$ matrices with non-negative entries with fixed specified row and column sums. For example, the \emph{Birkhoff polytope} is the set of all $n \times n$ matrices with non-negative entries whose rows and columns all sum to 1. Equivalently, the Birkhoff polytope is the set of all \emph{doubly stochastic} $n \times n$ matrices. All of these polytopes are linear slices of the positive orthant of some dimension.

Here, we generalize the notion of the Birkhoff polytope (and transportation polytopes more generally) to something which might be called the \emph{quantum Birkhoff polytope} or \emph{Birkhoff spectrahedron}. Specifically, fix $A \in \Sym(n^2)$, and define a linear map $\Phi_A: \Sym(n) \to \Sym(n)$ via
\[
    A = \sum_{i,j = 1}^n E_{i,j} \otimes \Phi_A(E_{i,j}),
\]
where $E_{i,j}$ is the matrix with a 1 in the $(i,j)$ entry and 0 elsewhere. Note that this defines $\Phi_A$ on a basis via the $n \times n$ blocks of the matrix $A$. By Choi's theorem \cite{choi1975completely}, we have that $A \in \PSD(n^2)$ if and only if $\Phi_A$ is a \emph{completely positive map}.\footnote{We do not define this term further here, but note that it is stronger than the condition $\Phi_A(\PSD(n)) \subseteq \PSD(n)$. E.g., see \cite{landau1993birkhoff} for further discussion.} This positivity condition on $A$ will serve to generalize the fact that elements of the Birkhoff polytope have non-negative entries. Also note that the adjoint linear operator $\Phi_A^*$ is given by
\[
    A = \sum_{i,j = 1}^n \Phi_A^*(E_{i,j}) \otimes E_{i,j},
\]
and thus Choi's theorem applies to $\Phi_A^*$ as well.

We now consider generalizations of the linear restrictions on the Birkhoff polytope; namely, the row and column sum restrictions. An equivalent description of these conditions is given by: for $B$ in the Birkhoff polytope of size $n$ we have
\[
    B \cdot 1_n = 1_n \qquad \text{and} \qquad B^* \cdot 1_n = 1_n,
\]
where $1_n$ is the length-$n$ all-ones column vector. Stated this way, these conditions are immediately generalizable to $A \in \PSD(n^2)$ via
\[
    \Phi_A(I_n) = I_n \qquad \text{and} \qquad \Phi_A^*(I_n) = I_n.
\]
The matrices $\Phi_A(I_n)$ and $\Phi_A^*(I_n)$ are also called the \textit{partial traces} of the operator $A$. For this reason, a completely positive map $\Phi_A$ (or sometimes $A$) is said to be \emph{doubly stochastic} when it satisfies the above conditions (e.g., see \cite{landau1993birkhoff} for further discussion). Note that the connection to the Birkhoff polytope is strengthened by the fact that the set of all diagonal $A \in \PSD(n^2)$ for which $\Phi_A$ is doubly stochastic is equal to the Birkhoff polytope after rearranging the diagonals of each $n \times n$ block of $A$ into columns of an $n \times n$ matrix.

\subsection{Multi-index transportation polytopes and completely positive maps}

The volume approximation result of \cite{bhgaussian} cannot be used directly to obtain asymptotics for the volume of the Birkhoff polytope and transportation polytopes. This comes from the fact that the asymptotic volume of the Birkhoff polytope given in \cite{canfield2007asymptotic} gives a different formula than the one achieved in \cite{bhgaussian}. That said, one can achieve the correct asymptotic formula using similar techniques and an ``Edgeworth correction'' term; see Remark \ref{rem:later_Birkhoff_work} for more discussion.

Because of this, Barvinok and Hartigan in \cite{bhgaussian} apply their results to what they call \emph{multi-index transportation polytopes}. Instead of considering matrices with fixed row and column sums, they consider higher-order tensors with fixed codimension-1 slice sums. More formally, they consider $n_1 \times n_2 \times \cdots \times n_k$ multi-dimensional matrices $A$ and fixed $\alpha_{ij} \in \Z_{\geq 0}$ for all $i \in [k]$ and $j \in [n_i]$, such that
\[
\sum_{\substack{\kappa \in [n_1] \times \cdots \times [n_k] \\ \kappa_i = j}} a_\kappa = \sum_{\kappa_1 = 1}^{n_1} \cdots \sum_{\kappa_{i-1} = 1}^{n_{i-1}} \sum_{\kappa_{i+1} = 1}^{n_{i+1}} \cdots \sum_{\kappa_k = 1}^{n_k} a_{\kappa_1,\ldots,\kappa_{i-1},j,\kappa_{i+1},\ldots,\kappa_k} = \alpha_{ij}
\]
for all $i \in [k]$ and $j \in [n_i]$. One can then define the \emph{multi-index Birkhoff polytope} by setting $n_1=\cdots=n_k=n$ and $\alpha_{ij} = 1$ for all $i \in [k]$ and $j \in [n]$. The results of \cite{bhgaussian} then apply to this case whenever $k \geq 5$.

This generalization from transportation polytopes to multi-index transportation polytopes then easily extends to completely positive maps. Fix $A \in \Sym(n_1n_2 \cdots n_k)$, and for all $i \in [k]$ define a linear map $\Phi_A^{(i)}: \Sym(n_1 \cdots n_{i-1}n_{i+1} \cdots n_k) \to \Sym(n_i)$ via
\[
    A = \sum_{\kappa,\kappa' \in \substack{[n_1] \times \cdots \times [n_{i-1}] \times \\ [n_{i+1}] \times \cdots \times [n_k]}} E_{\kappa_1,\kappa_1'} \otimes \cdots \otimes E_{\kappa_{i-1},\kappa_{i-1}'} \otimes \Phi_A^{(i)}(E_{\kappa,\kappa'}) \otimes E_{\kappa_i,\kappa_i'} \otimes \cdots \otimes E_{\kappa_{k-1},\kappa_{k-1}'},
\]
where $E_{\kappa,\kappa'} := E_{\kappa_1,\kappa_1'} \otimes \cdots \otimes E_{\kappa_{k-1},\kappa_{k-1}'}$. Note that as above, this expression defines $\Phi_A^{(i)}$ on a basis of $\Sym(n_1 \cdots n_{i-1}n_{i+1} \cdots n_k)$. By Choi's theorem \cite{choi1975completely} again, the following are then equivalent:
\begin{enumerate}
    \item $A \in \PSD(n_1n_2 \cdots n_k)$,
    \item $\Phi_A^{(i)}$ is a completely positive linear map for some $i \in [k]$,
    \item $\Phi_A^{(i)}$ is a completely positive linear map for all $i \in [k]$.
\end{enumerate}
With this, we now explicitly define the \emph{multi-index Birkhoff spectrahedron} which we denote $\mathcal{SCP}_{n,k}$ for ``Stochastic Completely Positive''. Setting $n_1=\cdots=n_k=n$ and given $A \in \PSD(n^k)$, we define:
\[
    A \in \mathcal{SCP}_{n,k} \qquad \iff \qquad \Phi_A^{(i)}(I_{n^{k-1}}) = I_n \quad \text{for all} \quad i \in [k].
\]
In particular, $\mathcal{SCP}_{n,2}$ is precisely the set of matrices $A \in \PSD(n^2)$ which correspond to doubly stochastic completely positive maps $\Phi_A = \Phi_A^{(2)}$, since $\Phi_A^* = \Phi_A^{(1)}$ in this case. 

As was the case with transportation polytopes and the Birkhoff polytope, our main result does not apply to $\mathcal{SCP}_{n,2}$. However, we will see next that our main result does yield asymptotics for the volume of $\mathcal{SCP}_{n,k}$ for all fixed $k \geq 7$.

\begin{remark} \label{rem:later_Birkhoff_work}
    The results of \cite{bhgaussian} only apply directly to multi-index transportation polytopes when the size of the tensor is $k \geq 5$. However, later work showed that similar techniques could handle the case of $k=2$ using an ``Edgeworth correction'' to the approximation formula (see \cite{barvinok2012asymptotic,barvinok2009maximum} for further discussion). Additionally, it was also later shown that the original approximation of \cite{bhgaussian} could be applied to $k=3,4$ cases in \cite{benson2014counting}.
\end{remark}

\subsection{Applying Theorem \ref{thm:asymptotic} to multi-index Birkhoff spectrahedra}

We now apply Theorem \ref{thm:asymptotic} to multi-index Birkhoff spectrahedra $\mathcal{SCP}_{n,k}$ which we defined above. To do this, we first rewrite the linear conditions which define $\mathcal{SCP}_{n,k}$ in a way which is more compatible with the statement of Theorem \ref{thm:asymptotic}.


%
%
\paragraph{Linear constraints.} For $i \in [k]$ and $\alpha,\alpha' \in [n]$ we define
\[
    A_{i,\alpha,\alpha'} := \sum_{\substack{\kappa,\kappa' \in [n]^k \\ \kappa_i = \alpha,\, \kappa_i' = \alpha' \\ \kappa_j = \kappa_j' \, \forall j \neq i}} \frac{E_{\kappa,\kappa'} + E_{\kappa',\kappa}}{2} \qquad \text{and} \qquad b_{i,\alpha,\alpha'} = \delta_{\alpha,\alpha'},
\]
where $E_{\kappa,\kappa'} := E_{\kappa_1,\kappa_1'} \otimes \cdots \otimes E_{\kappa_k,\kappa_k'}$ and $E_{p,q}$ is the matrix with a 1 in the $(p,q)$ entry and 0 elsewhere. An alternative description of $\mathcal{SCP}_{n,k}$ is then given by
\[
    \mathcal{SCP}_{n,k} = \{X \in \PSD(n^k) : (\tr(A_{i,\alpha,\alpha}X))_{i \in [k],\alpha, \alpha' \in [n], \alpha \leq \alpha'} = b\}
\]
where $\alpha,\alpha'$ range over $[n]$ with $\alpha \leq \alpha'$ and $i$ ranges over $[k]$. Further note that this description is partially redundant. By including an extra constraint on the trace of a given $X \in \PSD(n^k)$, we can eliminate the conditions indexed by $(i,n,n)$ for all $i \in [k]$. Specifically, we add the condition
\[
    \tr(A_0 X) = b_0 \qquad \text{where} \qquad A_0 = \frac{I_{n^k}}{n} \qquad \text{and} \qquad b_0 = 1.
\]
Thus the total number of affine constraints required to describe $\mathcal{SCP}_{n,k}$ is $m_n = k \binom{n+1}{2} - k + 1$.

\paragraph{Entropy function maximizer.} We next claim that $P^\star := \frac{I_{n^k}}{n^{k-1}}$ maximizes $\phi(P)$ over all $P \in \mathcal{SCP}_{n,k}$, or equivalently maximizes $\log \det P$. To see this, recall that the gradient of $\log \det P$ is $P^{-1}$ since we are restricting to real symmetric matrices. For any $P \in \mathcal{SCP}_{n,k}$, we then have
\[
    \langle P - P^\star, \nabla \phi(P^\star) \rangle = \tr(P(P^\star)^{-1}) - n^k = n^{k-1} \tr(P) - n^k = 0,
\]
since the condition $\tr(A_0 P) = b_0$ implies $\tr(P) = n$. Since $\phi(P)$ is concave, this implies $P^\star = \frac{I_{n^k}}{n^{k-1}}$ maximizes $\phi$ over $\mathcal{SCP}_{n,k}$.

\paragraph{Applying Theorem \ref{thm:asymptotic}.} To apply Theorem \ref{thm:asymptotic}, we need to determine for which values of $k$ the number of affine constraints $m_n = k \binom{n+1}{2} - k + 1$ is sufficiently dominated by $N_n = n^k$. Recalling Equation \ref{eq:main5} from the statement of Theorem \ref{thm:asymptotic}, we compute
\[
    \lim_{n \to \infty} \frac{m_n^3 \log N_n}{N_n} = \lim_{n \to \infty} \frac{\left(k\binom{n+1}{2} - k + 1\right)^3 k \log n}{n^k} \leq \lim_{n \to \infty} \frac{k^4 n^6 \log n}{n^k}.
\]
Thus whenever $k \geq 7$, we can apply Theorem \ref{thm:asymptotic} to obtain asymptotics for the volume of $\mathcal{SCP}_{n,k}$.

\paragraph{The asymptotic volume.} We now compute the asymptotic volume of $\mathcal{SCP}_{n,k}$ for $k \geq 7$. For $N = n^k$ and $m = k\binom{n+1}{2} - k + 1$, we have
\[
    \vol(\mathcal{SCP}_{n,k}) \approx \left(\frac{N+1}{4\pi}\right)^{m/2} \left(\frac{\det(AA^\top)}{\det(BB^\top)}\right)^{1/2} e^{\phi(P^\star)},
\]
where $\left(\frac{\det(AA^\top)}{\det(BB^\top)}\right)^{1/2} = n^{m(k-1)} = \left(\frac{N}{n}\right)^m$ since $B = \frac{1}{n^{k-1}} A$, and
\[
\begin{split}
    e^{\phi(P^\star)} &= \left(\frac{N+1}{2e}\right)^{-\frac{N(N+1)}{2}} \left(n^{N(k-1)}\right)^{-\frac{N+1}{2}} \Gamma_N\left(\frac{N+1}{2}\right) \\
        &= \left(\frac{2en}{N(N+1)}\right)^{\frac{N(N+1)}{2}} \Gamma_N\left(\frac{N+1}{2}\right)
\end{split}
\]
since $P^\star = \frac{I_{n^k}}{n^{k-1}}$. That is,
\[
    \vol(\mathcal{SCP}_{n,k}) \approx \left(\frac{N+1}{4\pi}\right)^{\frac{m}{2}} \left(\frac{N}{n}\right)^m \left(\frac{2en}{N(N+1)}\right)^{\frac{N(N+1)}{2}} \Gamma_N\left(\frac{N+1}{2}\right).
\]
Further, using the definition of the multivariate gamma function we obtain
\[
    \Gamma_N\left(\frac{N+1}{2}\right) = \left(\frac{\pi}{2}\right)^{\frac{N(N-1)}{4} + \frac{1}{2} \lfloor\frac{N}{2}\rfloor} \prod_{j=0}^{N-1} j!.
\]


\section{Proof of the Main Technical Result} \label{sec:proof_main}



Our main results in Section \ref{sec:main_result} follow from a single result, which we prove in this section. Although this result is strictly stronger than any of the results in Section \ref{sec:main_result}, we have moved it here because it is much more technical and complicated to state. The result we prove here is the direct analog of the main result of \cite{bhgaussian}.

We first recall all notation from Section \ref{sec:main_result}. Given $A_1,A_2,\dots,A_m\in\Sym(N)$ and $b\in\R^m$, we define a spectrahedron $\mathcal{S}$ by:
\[
    \mathcal{S} := \big\{P \in \PSD(N) \; : \; \tr(A_k P) = b_k \quad \text{for} \quad k \in [m]\big\}.
\]
We assume that $\mathcal{S}$ is compact, that the constraints $\tr(A_k P) = b_k$ are linearly independent, that $m < \binom{N+1}{2} = \dim(\PSD(N))$, and that $\mathcal{S}$ is of dimension exactly $\binom{N+1}{2} - m$. Let $P^\star \in \mathcal{S}$ be the point which maximizes the function 
\[
\begin{split}
    \phi(P) &= \log\Gamma_N\left(\frac{N+1}{2}\right) - \frac{N(N+1)}{2} \log\left(\frac{N+1}{2e}\right) + \frac{N+1}{2}\log\det(P) \\
        &= \text{const}(N) + \frac{N+1}{2}\log\det(P)
\end{split}
\]
over $\mathcal{S}$.
Let $A$ and $B$ be linear operators from $\Sym(N)$ to $\R^m$, defined via
\[
    AX := (\tr(A_1X), \ldots, \tr(A_mX))
\]
and
\[
    BX := (\tr(Z_1X),\ldots,\tr(Z_mX)).
\]
where $Z_k := \sqrt{P^\star} A_k \sqrt{P^\star}$ for all $k \in [m]$.

\begin{theorem}[Main technical result] \label{thm:main}
    Let $\mathcal{S}$ be a spectrahedron as defined above. Consider the quadratic form $q: \R^m \to \R$ defined by
    \[
        q(t) := \frac{1}{N+1} \tr\left[\left(\sum_{k=1}^m t_k Z_k\right)^2\right].
    \]
    Suppose that for some $\lambda > 0$ we have
    \begin{align} \label{eq:lambda_assumption} \tag{A1}
        q(t) \geq \lambda \|t\|_2^2 \quad \text{for all} \quad t \in \R^m,
    \end{align}
    and that for some $\theta > 0$ we have
    \begin{align} \label{eq:theta_assumption} \tag{A2}
        \frac{2}{N+1} \|(x^\top Z_k x)_{k=1}^m\|_2 \leq \theta \quad \text{for all} \quad \|x\|_2 = 1.
    \end{align}
    Then there exists an absolute constant $\gamma$ (we can choose $\gamma = 10^5$) such that the following holds:
    
    Fix $0 < \epsilon \leq \frac{1}{2}$ and suppose that
    \begin{align} \label{eq:gamma_assumption} \tag{A3}
        \lambda \geq \gamma \theta^2 \epsilon^{-2} m \left(m + \log(\epsilon^{-1})\right)^2 \log(N\epsilon^{-1}).
    \end{align}
    Then the number
    \[
        \left(\frac{N+1}{4\pi}\right)^{m/2} \left(\frac{\det(AA^\top)}{\det(BB^\top)}\right)^{1/2} e^{\phi(P^\star)}
    \]
    approximates $\vol(\mathcal{S})$ within relative error $\epsilon$.
\end{theorem}

The remainder of this section is devoted to proving this result.

\paragraph{Setup.} Let $X$ be a random variable distributed according to the (maximum entropy) Wishart distribution with expectation given by a positive definite matrix $P^\star$. That is,
\[
    X = \sqrt{P^*}\, \frac{GG^T}{N+1}\,\sqrt{P^*}
\] where $G\in\R^{N\times (N+1)}$ is a random matrix where the entries $G_{ij}\sim\mathcal{N}(0,1)$ are drawn identically and independently from the normal distribution.

Given $N \times N$ real symmetric matrices $A_1,\ldots,A_m$, we further define
\[
    Y := AX := (\tr(A_1X), \ldots, \tr(A_mX)) \in \mathbb{R}^m,
\]
and we note that
\[
    b = AP^\star = (\tr(A_1P^\star), \ldots, \tr(A_mP^\star)) \in \mathbb{R}^m.
\]
We further define
the matrix
\[
    Z(t) := \sum_{k=1}^m t_k Z_k = \sum_{k=1}^m t_k \sqrt{P^\star} A_k \sqrt{P^\star},
\]
and with this, $q(t)$ can be defined via
\[
    q(t) := \frac{1}{N+1} \tr(Z^2(t)) = \frac{1}{N+1} \sum_{j=1}^N \lambda_j^2(Z(t))
\]
for any $t \in \mathbb{R}^m$, where $\lambda_j(M)$ denotes the $j^\text{th}$ largest eigenvalue of $M$. We show how this form relates to $B$ in Lemma \ref{lem:q_B_expression}. Note further that we have slightly overloaded the symbol $\lambda$: $\lambda$ itself denotes the above constant, and $\lambda_j(M)$ denotes the $j^\text{th}$ largest eigenvalue of $M$. We have done this mainly to retain consistency of notation between our computations and that of Barvinok and Hartigan in \cite{bhgaussian}.

\paragraph{The proof.} We now define $\sigma := 4m + 10 \log \frac{1}{\epsilon}$, and we further define
\[
    R_1 := \left\{t : \|t\|_2 \leq \frac{1}{2\theta},~ q(t) > \sigma\right\} \quad \text{and} \quad R_2 := \left\{t : q(t) \leq \sigma\right\} \quad \text{and} \quad R_3 := \left\{t : \|t\|_2 \geq \frac{1}{2\theta}\right\}.
\]

Note that for $t$ with $\Vert t\Vert_2\geq \frac{1}{2\theta}$ and for $\gamma> 40$, we have \[
q(t)\geq \frac{\lambda}{4\theta^2}\geq\frac{\gamma\epsilon^{-2}m(m+\log\frac{1}{\epsilon})^2\log(N\epsilon^{-1})}{4}\geq \frac{\gamma\sigma m (m+\log\frac{1}{\epsilon})\log(N\epsilon^{-1})}{40}>\sigma.
\] Thus, the sets $R_1,R_2,R_3$ are disjoint for large enough $\gamma$.

We assume the following bound, which we prove in Section \ref{sec:large_q} (for the $R_1$ bound), Section \ref{sec:small_q} (for the $R_2$ bound), and Section \ref{sec:large_t} (for the $R_3$ bound):
\[
\begin{split}
    \left|\int_{\mathbb{R}^m} e^{-i \langle b, t \rangle} \phi_Y(t) \diff t - \int_{\mathbb{R}^m} e^{-q(t)} \diff t\right| &\leq \left|\int_{R_1} e^{-i \langle b, t \rangle} \phi_Y(t) \diff t\right| \\
        &+ \left|\int_{R_2} e^{-i \langle b, t \rangle} \phi_Y(t) \diff t - \int_{\mathbb{R}^m} e^{-q(t)} \diff t\right| \\
        &+ \left|\int_{R_3} e^{-i \langle b, t \rangle} \phi_Y(t) \diff t\right| \\
        &\leq \left[\epsilon^3 + \left(\frac{2\epsilon}{3} + \epsilon^5\right) + \frac{\epsilon}{100}\right] \int_{\mathbb{R}^m} e^{-q(t)} \diff t \\
        &\leq \epsilon \int_{\mathbb{R}^m} e^{-q(t)} \diff t,
\end{split}
\]
where $\phi_Y$ is the characteristic function of the random variable $Y$, the first inequality follows from the triangle inequality, and the last inequality follows from the fact that $\epsilon \leq \frac{1}{2}$. Assuming this bound, we now complete the proof of the main result Theorem \ref{thm:main}. By the characteristic function inversion formula (e.g., see Section 29 of \cite{billingsley}), the density of $Y$ at $b$ is equal to
\[
    \frac{1}{(2\pi)^m} \int_{\mathbb{R}^m} e^{-i \langle b, t \rangle} \phi_Y(t) \diff t.
\]
By Corollary \ref{cor:density_formula} and since $A$ is a linear map, the density of $Y = AX$ at $b$ is also equal to
\[
    \frac{\vol(\mathcal{S}) e^{-\phi(P^\star)}}{\det(AA^\top)^{1/2}}.
\]
Further, using Lemma \ref{lem:q_B_expression} below, we compute the multivariate Gaussian integral
\[
    \int_{\mathbb{R}^m} e^{-q(t)} \diff t = \int_{\mathbb{R}^m} e^{-\frac{1}{2} t^\top \left(\frac{2}{N+1}BB^\top\right)t} \diff t = \sqrt{\frac{(2\pi)^m}{\det\left(\frac{2}{N+1}BB^\top\right)}} = \frac{((N+1)\pi)^{m/2}}{\det(BB^\top)^{1/2}}.
\]
Combining everything and rearranging then implies
\[
    \left|\frac{\vol(\mathcal{S})}{\left(\frac{N+1}{4\pi}\right)^{m/2} \left(\frac{\det(AA^\top)}{\det(BB^\top)}\right)^{1/2} e^{\phi(P^\star)}} - 1\right| \leq \epsilon.
\]
That is,
\[
    \left(\frac{N+1}{4\pi}\right)^{m/2} \left(\frac{\det(AA^\top)}{\det(BB^\top)}\right)^{1/2} e^{\phi(P^\star)} \quad \text{approximates} \quad \vol(\mathcal{S}) \quad \text{within relative error} \quad \epsilon,
\]
which is precisely the statement of Theorem \ref{thm:main}. The remainder of this section will now be spent proving the bound which we assumed above.

\begin{remark}
    This technique of bounding the characteristic function is precisely what was used in \cite{bhgaussian} in the polytope case. See \cite{barvinok2022system} for another interesting use of these types of integral expressions involving real symmetric matrices, where similar computations are utilized on the algebraic problem of solving a system of real quadratic equations.
\end{remark}

\subsection{A few small results}
For a symmetric matrix $Z\in \Sym(N)$, we denote by $|Z|$ the matrix $\sqrt{Z^2}$.
\begin{lemma} \label{lem:Z_max_eig}
    For $Z(t)$ defined as above and any $t \in \mathbb{R}^m$, we have
    \[
        \lambda_{\max}(|Z(t)|) \leq \frac{N+1}{2} \theta \|t\|_2.
    \]
\end{lemma}
\begin{proof}
    Since $Z(t)$ is 
    symmetric, we have
    \[
    \begin{split}
        \lambda_{\max}(Z(t)) &= \sup_{\|x\|_2=1} x^\top \left(\sum_{k=1}^m t_k Z_k\right) x = \sup_{\|x\|_2=1} \sum_{k=1}^m t_k x^\top Z_k x \\
            &\leq \sup_{\|x\|_2=1} \sum_{k=1}^m \left|t_k x^\top Z_k x\right| \leq \|t\|_2 \cdot \sup_{\|x\|_2 = 1} \|(x^\top Z_k x)_{k=1}^m\|_2.
    \end{split}
    \]
    The same holds for $\lambda_{\max}(-Z(t))$, and so the bound in fact holds for $\lambda_{\max}(|Z(t)|)$. Applying the $\theta$ bound of (\ref{eq:theta_assumption}) then gives the result.
\end{proof}

\begin{corollary} \label{cor:tr_Z_p_bound}
    For $Z(t)$ defined as above and any $t \in \mathbb{R}^m$, we have
    \[
        \tr(|Z^p(t)|) \leq \left(\frac{N+1}{2} \theta \|t\|_2\right)^{p-2} \tr(Z^2(t)).
    \]
\end{corollary}
\begin{proof}
    For any $p$, we have
    \[
    \begin{split}
        \tr(|Z^p(t)|) &= \sum_{j=1}^N \lambda_j^p(|Z(t)|) \leq \lambda_{\max}(|Z(t)|)^{p-2} \sum_{j=1}^N \lambda_j^2(|Z(t)|) = \lambda_{\max}(|Z(t)|)^{p-2} \tr(Z^2(t)).
    \end{split}
    \]
    Applying the previous lemma then gives the result.
\end{proof}

\begin{lemma} \label{lem:q_B_expression}
    For any $t \in \mathbb{R}^m$, we have
    \[
        q(t) = \frac{1}{N+1} t^\top B B^\top t.
    \]
\end{lemma}
\begin{proof}
    Letting $\mathrm{vec}$ denote the standard vectorization, we can write
    \[
        BX = \begin{bmatrix}
            | &  & | \\
            \mathrm{vec}(Z_1) & \cdots & \mathrm{vec}(Z_m) \\
            | &  & | \\
        \end{bmatrix}^\top \mathrm{vec}(X),
    \]
    which implies
    \[
        B^\top t = \mathrm{vec}^{-1}\left(
        \begin{bmatrix}
            | &  & | \\
            \mathrm{vec}(Z_1) & \cdots & \mathrm{vec}(Z_m) \\
            | &  & | \\
        \end{bmatrix}
        t \right) = \sum_{k=1}^m t_k Z_k = Z(t).
    \]
    Therefore,
    \[
        \frac{1}{N+1} t^\top B B^\top t = \frac{1}{N+1} \tr(Z^2(t)) = q(t).
    \]
\end{proof}

We now state few results from \cite{bhgaussian} which we will need here.

\begin{lemma}[\cite{bhgaussian}, Lemma 6.2] \label{lem:q_bound_int}
    Let $q: \mathbb{R}^m \to \mathbb{R}$ be a positive definite quadratic form, and let $\omega > 0$ be a positive real number.
    \begin{enumerate}
        \item If $\omega \geq 3$, then
        \[
            \int_{t:~q(t) \geq \omega m} e^{-q(t)} \diff t \leq e^{-\omega m/2} \int_{\mathbb{R}^m} e^{-q(t)} \diff t.
        \]
        \item If $q(t) \geq \lambda \|t\|_2^2$ for all $t \in \mathbb{R}^m$, then for any $a \in \mathbb{R}^m$ we have
        \[
            \int_{t:~|\langle a, t \rangle| > \omega \|a\|_2} e^{-q(t)} \diff t \leq e^{-\lambda \omega^2} \int_{\mathbb{R}^m} e^{-q(t)} \diff t.
        \]
    \end{enumerate}
\end{lemma}

\begin{lemma}[\cite{bhgaussian}, Lemma 6.3] \label{lem:t_bound_int}
    For any $\rho \geq 0$ and any $k > m$ where $t \in \R^m$, we have
    \[
        \int_{t:~\|t\|_2 \geq \rho/\theta} (1 + \theta^2\|t\|_2^2)^{-k/2} \diff t \leq \frac{2\pi^{m/2}}{\Gamma(m/2) \cdot \theta^m(k-m)} (1+\rho^2)^{(m-k)/2}.
    \]
\end{lemma}

We actually need a slight modification of part(2) of Lemma~\ref{lem:q_bound_int} for our purposes, which we prove now.
\begin{corollary}
\label{cor:q_bound_int_ball_but_better}
 Let $q: \mathbb{R}^m \to \mathbb{R}$ be a positive definite quadratic form. If $q(t) \geq \lambda \|t\|_2^2$ for all $t \in \mathbb{R}^m$, then for every positive real number $\omega\geq \sqrt{\frac{3m}{\lambda}}$ 
    \[
        \int_{\|t\|_2 \geq \omega} e^{-q(t)} \diff t \leq  e^{-\lambda \omega^2/2} \int_{\mathbb{R}^m} e^{-q(t)} \diff t
    \] holds.
    
\end{corollary}
\begin{proof}
    For $t\in\R^m$ with $\Vert t\Vert_2\geq\omega$, we have $q(t)\geq \lambda\omega^2$. Moreover, the condition $\omega\geq\sqrt{\frac{3m}{\lambda}}$ implies $\lambda\omega^2\geq 3m$. Hence, \[
    \int_{\|t\|_2 \geq \omega} e^{-q(t)} \diff t \leq \int_{q(t) \geq \lambda\omega^2} e^{-q(t)} \diff t \leq e^{-\lambda\omega^2/2} \int_{\R^m} e^{-q(t)} \diff t 
    \] where the second inequality follows from Lemma~\ref{lem:q_bound_int}~(1).
\end{proof}

\subsection{The characteristic function for small $t$} \label{sec:char_small_t}

In this and the following section, we will bound the characteristic function of the random variable $Y$ for various input values. To do this, we apply the appropriate transformation to the characteristic function of the maximum entropy Wishart random variable $X$. The expectation of $X$ is $P^\star$, and thus its characteristic function is given by
\[
    \phi_X(M) = \det\left[I - \frac{2i\sqrt{P^\star}M\sqrt{P^\star}}{N+1}\right]^{-\frac{N+1}{2}}.
\]
Note that there is a potential ambiguity in this definition coming from applying the square root to a complex number, and we briefly address this now. Considering the determinant expression above as a function on the space of complex symmetric matrices, the expression is non-zero in an open neighborhood of the subspace of real symmetric matrices. (The eigenvalues of the input matrix will all be near the line $\Re(z) = 1$.) Thus the square root of the determinant can be defined analytically in an open neighborhood up to choice of square root, and we make the choice which gives $\phi_X(0) = 1$. Because all eigenvalues of the matrix in the expression for $\phi_X(M)$ are contained in the open right half-plane whenever $M$ is real symmetric, the value of $\phi_X$ is also given by applying the square root to the eigenvalues individually, using the principal branch. That is, we have
\[
    \phi_X(M) = \left(\prod_{j=1}^N \left[1-2i \cdot \lambda_j\left(\frac{\sqrt{P^\star} M \sqrt{P^\star}}{N+1}\right)\right]^{-\frac{1}{2}}\right)^{N+1}.
\]
We now use this expression for for the characteristic function of $X$ to prove a nice expression for the characteristic function of $Y$.

\begin{lemma} \label{lem:char_func_expression_small_t}
    On the set of all real $t \in \mathbb{R}^m$ such that $\|t\|_2 \leq \frac{1}{2\theta}$, the characteristic function of $Y$ can be expressed as
    \[
        \phi_Y(t) = \exp(i \langle b, t \rangle - q(t) - i f(t) + g(t))
    \]
    where
    \[
        f(t) = \frac{4}{3(N+1)^2} \cdot \tr(Z^3(t)) \quad \text{and} \quad |g(t)| \leq \frac{4}{(N+1)^3} \cdot \tr(Z^4(t)).
    \]
\end{lemma}
\begin{proof}
    Using the above discussion, the expression for the characteristic function of a random $N \times N$ matrix $X$ distributed according to the maximum entropy Wishart distribution with expectation $P^\star$ is given by
    \[
        \log\phi_X(M) = (N+1) \sum_{j=1}^N -\frac{1}{2} \log\left[1-2i \cdot \lambda_j\left(\frac{\sqrt{P^\star} M \sqrt{P^\star}}{N+1}\right)\right],
    \]
    where here we choose the principal branch of $\log$ as discussed above.
    
    We now write down the characteristic function for the random variable $Y = AX$, where $A$ acts by $X \mapsto (\tr(A_iX))_{i=1}^m$ as defined above. Using the standard formula for the characteristic function under the action of a linear operator, we have
    \[
        \log\phi_Y(t) = \log\phi_X(A^\top t) = -\frac{N+1}{2} \sum_{j=1}^N \log\left[1-2i \cdot \lambda_j\left(\frac{\sqrt{P^\star} (A^\top t) \sqrt{P^\star}}{N+1}\right)\right].
    \]
    Note that the above expression is syntactically valid, since we can view $A^\top$ as a map from $\R^m$ to $N \times N$ real symmetric matrices. Letting $\mathrm{vec}$ denote the standard vectorization, we can write
    \[
        Y = AX = \begin{bmatrix}
            | & | &  & | \\
            \mathrm{vec}(A_1) & \mathrm{vec}(A_2) & \cdots & \mathrm{vec}(A_m) \\
            | & | &  & | \\
        \end{bmatrix}^\top \mathrm{vec}(X),
    \]
    which implies
    \[
        A^\top t = \mathrm{vec}^{-1}\left(
        \begin{bmatrix}
            | & | &  & | \\
            \mathrm{vec}(A_1) & \mathrm{vec}(A_2) & \cdots & \mathrm{vec}(A_m) \\
            | & | &  & | \\
        \end{bmatrix}
        t \right) = \sum_{k=1}^m t_k A_k.
    \]
    Therefore we have
    \[
        \log\phi_Y(t) = -\frac{N+1}{2} \sum_{j=1}^N \log\left(1 - \frac{2i}{N+1} \lambda_j(Z(t))\right).
    \]
    Note that since $\|t\|_2 \leq \frac{1}{2\theta}$ by assumption, we have $\lambda_{\max}(|Z(t)|) \leq \frac{N+1}{4}$ by Lemma \ref{lem:Z_max_eig}. Now recall the Taylor's approximation $\log(1+\xi)=\sum_{i=1}^n (-1)^{i-1}\, \xi^i/i$. Using the error theorem for Taylor's approximation, for every $\xi\in\mathbb{C}$ with $|\xi|\leq \frac{1}{2}$ there exists $\tilde{\xi}$ with $|\tilde{\xi}|\leq \frac{1}{2}$ such that 
    \[
    \log(1+\xi) - \xi +\frac{\xi^2}{2} - \frac{\xi^3}{3} = \frac{\xi^4}{4(1+\tilde{\xi})}
    \]
    which implies
    \[
        \log(1+\xi) = \xi - \frac{\xi^2}{2} + \frac{\xi^3}{3} + z_0 \cdot \xi^4
    \]
    for some $|z_0| \leq \frac{1}{2}$. Therefore for any fixed $j$, we have
    \[
    \begin{split}
        \log\left(1 - \frac{2i}{N+1} \lambda_j(Z(t))\right) &= -\frac{2i}{N+1} \lambda_j(Z(t)) + \frac{2}{(N+1)^2} \lambda_j^2(Z(t)) \\
            &+ \frac{8i}{3(N+1)^3} \lambda_j^3(Z(t)) + \hat{g}_j(t) \cdot \frac{16}{(N+1)^4} \lambda_j^4(Z(t)),
    \end{split}
    \]
    where $|\hat{g}_j(t)| \leq \frac{1}{2}$. Since $\sum_{j=1}^N \lambda_j^p(Z(t)) = \tr(Z^p(t))$, we have
    \[
    \begin{split}
        \log\phi_Y(t) &= i \cdot \tr(Z(t)) - \frac{1}{N+1} \cdot \tr(Z^2(t)) \\
            &- \frac{4i}{3(N+1)^2} \cdot \tr(Z^3(t)) + \hat{g}(t) \cdot \frac{8}{(N+1)^3} \cdot \tr(Z^4(t))
    \end{split}
    \]
    where $|\hat{g}(t)| \leq \frac{1}{2}$. Note that this relies on the fact that $\lambda_j^4(Z(t)) \geq 0$ for all $j$. The fact that
    \[
        \tr(Z(t)) = \sum_{k=1}^m t_k \tr(\sqrt{P^\star} A_k \sqrt{P^\star}) = \sum_{k=1}^m t_k b_k = \langle b, t \rangle
    \]
    then implies the result.
\end{proof}

\begin{corollary} \label{cor:char_func_bound_small_t}
    For real $t \in \mathbb{R}^m$ such that $\|t\|_2 \leq \frac{1}{2\theta}$, the characteristic function of $Y$ can be bounded by
    \[
        |\phi_Y(t)| \leq e^{-3q(t)/4}.
    \]
\end{corollary}
\begin{proof}
    By the previous lemma and Corollary \ref{cor:tr_Z_p_bound}, we have 
    \[
        |\phi_Y(t)| = \exp(-q(t) + \Re[g(t)])
    \]
    where $\Re[g(t)]$ denotes the real part of $g(t)$ and
    \[
    \begin{split}
        |g(t)| &\leq \frac{4}{(N+1)^3} \tr(Z^4(t)) \leq \frac{4}{(N+1)^3} \left(\frac{N+1}{2} \theta \|t\|_2\right)^2 \tr(Z^2(t)) \leq \frac{1}{4(N+1)} \tr(Z^2(t)) = \frac{q(t)}{4}.
    \end{split}
    \]
    This gives the desired bound.
\end{proof}

\subsubsection{For large $q(t)$} 
\label{sec:large_q}
Defining $\sigma := 4m + 10 \log \frac{1}{\epsilon}$ as above, we consider the subcase where $q(t) > \sigma$ and $\|t\|_2 \leq \frac{1}{2\theta}$. By Corollary \ref{cor:char_func_bound_small_t} and part $(1)$ of Lemma \ref{lem:q_bound_int}, we have
\[
\begin{split}
    \left|\int_{\substack{\|t\|_2 \leq \frac{1}{2\theta} \\ q(t) > \sigma}} e^{-i \langle b, t \rangle} \phi_Y(t) \diff t\right| &\leq \int_{3q(t)/4 > 3\sigma/4} e^{-3q(t)/4} \diff t \\
        &\leq e^{-3\sigma/8} \int_{\mathbb{R}^m} e^{-3q(t)/4} \diff t \\
        &\leq e^{-3m/2 -15/4 \log(1/\epsilon)}\int_{\mathbb{R}^m} e^{-3q(t)/4} \diff t \\
        &\leq e^{-3m/2} \epsilon^3 \int_{\mathbb{R}^m} e^{-3q(t)/4} \diff t.
\end{split}
\]
Since $q$ is a quadratic form, we further have
\[
    \int_{\mathbb{R}^m} e^{-3q(t)/4} \diff t = \int_{\mathbb{R}^m} e^{-q(t\sqrt{3/4})} \diff t = \left(\frac{4}{3}\right)^{m/2} \int_{\mathbb{R}^m} e^{-q(t)} \diff t.
\]
Since $e^{-3m/2} \left(\frac{4}{3}\right)^{m/2} = \exp(-(m/2) \cdot (3 - \log \frac{4}{3})) \leq 1$, this implies
\[
    \left|\int_{\substack{\|t\|_2 \leq \frac{1}{2\theta} \\ q(t) > \sigma}} e^{-i\langle b,t\rangle}\,\phi_Y(t) \diff t\right| \leq \epsilon^3 \int_{\mathbb{R}^m} e^{-q(t)} \diff t.
\]

\subsubsection{For small $q(t)$} \label{sec:small_q}

Defining $\sigma := 4m + 10 \log \frac{1}{\epsilon}$ as above, we now consider the case where $q(t) \leq \sigma$. We first show that this implies $\|t\|_2 \leq \frac{1}{2\theta}$ for $\gamma \geq 40$, meaning that we may consider this as a subcase of the $\|t\|_2 \leq \frac{1}{2\theta}$ case. Since $q(t) \leq \sigma$, we have that $\|t\|_2^2 \leq \frac{\sigma}{\lambda}$ by the assumption given in (\ref{eq:lambda_assumption}). Then by the assumption given in (\ref{eq:gamma_assumption}), we have 
\[
\begin{split}
    \|t\|_2^2 \leq \frac{\sigma}{\lambda} &\leq \frac{4m + 10 \log(\epsilon^{-1})}{\gamma \theta^2 \epsilon^{-2} m (m+\log(\epsilon^{-1}))^2 \log(N\epsilon^{-1})} \\
        &\leq \frac{10(m + \log(\epsilon^{-1}))}{40\theta^2(m + \log(\epsilon^{-1}))} \leq \frac{1}{4\theta^2}.
\end{split}
\]
With this, Lemma \ref{lem:char_func_expression_small_t} implies
\[
    \left|\int_{q(t) \leq \sigma} e^{-i \langle b, t \rangle} \phi_Y(t) \diff t - \int_{q(t) \leq \sigma} e^{-q(t)} \diff t\right| \leq \int_{q(t) \leq \sigma} e^{-q(t)} \left|e^{-i f(t) + g(t)} - 1\right| \diff t.
\]
And by Corollary \ref{cor:tr_Z_p_bound}, we further have
\[
    |g(t)| \leq \frac{4}{(N+1)^3} \cdot \tr(Z^4(t)) \leq \frac{4}{(N+1)^3} \left(\frac{N+1}{2}\theta\|t\|_2\right)^2 \tr(Z^2(t)) = \theta^2 \|t\|_2^2 \cdot q(t) \leq \frac{\theta^2 \sigma^2}{\lambda}.
\]
For $\gamma \geq 500$, our assumption on $\lambda$ given in (\ref{eq:gamma_assumption}) then implies
\[
    |g(t)| \leq \frac{\theta^2 \sigma^2}{\lambda} \leq \frac{\epsilon^2 (4m + 10 \log \epsilon^{-1})^2}{\gamma m (m + \log(\epsilon^{-1}))^2 \log(N\epsilon^{-1})} \leq \frac{100\epsilon^2}{\gamma} \leq \frac{\epsilon}{10},
\]
since $\epsilon \leq \frac{1}{2}$. Now define
\[
    T := \{t : q(t) \leq \sigma\} \qquad \text{and} \qquad B := \left\{t : \|t\|_2 \leq \frac{\epsilon}{10\sigma\theta}\right\}.
\]
Thus for $t \in T$, we have that
\[
    \left|e^{-i f(t) + g(t)} - 1\right| \leq e^{|g(t)|} + 1 \leq e^{\epsilon/10} + 1 \leq 3.
\]
Note that for $\gamma\geq 30000$, \[
\frac{\epsilon^2}{100\sigma^2\theta^2} \geq \frac{\gamma m (m+\log\epsilon^{-1})^2\log(N\epsilon^{-1})}{100\sigma^2\lambda}\geq \frac{m \log(N\epsilon^{-1})}{10^4 \lambda}\geq \frac{3m}{\lambda}
\] holds. Hence, by Corollary \ref{cor:q_bound_int_ball_but_better}, we have
\[
\int_{\mathbb{R}^m \setminus B} e^{-q(t)} \diff t \leq e^{-\frac{\lambda\epsilon^2}{200\sigma^2\theta^2}} \int_{\mathbb{R}^m} e^{-q(t)} \diff t.
\]
For $\gamma \geq 100000$, our assumption on $\lambda$ given in (\ref{eq:gamma_assumption}) then implies
\[
\begin{split}
    - \frac{\lambda\epsilon^2}{200\sigma^2\theta^2} &\leq- \frac{\gamma m(m+\log(\epsilon^{-1}))^2\log(N\epsilon^{-1})}{200(4m + 10\log(\epsilon^{-1}))^2} \\
        &\leq -\frac{\gamma \log(N\epsilon^{-1})}{20000} \\
        &\leq \log(\epsilon^5 N^{-5}) \leq \log \frac{\epsilon}{16}
\end{split}
\]
sinc $\epsilon \leq \frac{1}{2}$. Combining the above two expressions gives
\[
    \int_{\mathbb{R}^m \setminus B} e^{-q(t)} \diff t \leq \frac{\epsilon}{16} \int_{\mathbb{R}^m} e^{-q(t)} \diff t.
\]
For $t \in T \cap B$, Lemma \ref{lem:char_func_expression_small_t} and Corollary \ref{cor:tr_Z_p_bound} then imply
\[
    |f(t)| \leq \frac{4}{3(N+1)^2} \cdot \tr(|Z^3(t)|) \leq \frac{4}{3(N+1)^2} \cdot \frac{N+1}{2} \theta \|t\|_2 \cdot \tr(Z^2(t)) \leq \frac{2 \theta \|t\|_2}{3} \cdot q(t) \leq \frac{\epsilon}{15}
\]
and
\[
    |g(t)| \leq \frac{4}{(N+1)^3} \cdot \tr(Z^4(t)) \leq \theta^2 \|t\|_2^2 \cdot q(t) \leq \left(\frac{\epsilon}{10\sigma}\right)^2 \cdot \sigma \leq \frac{\epsilon^2}{100}.
\]
For $|x| < 1$, we have that $|\partial_x e^x| \leq e$. Thus for $t \in T \cap B$, we have that 
\[
    \left|e^{-i f(t) + g(t)} - 1\right| \leq \left|e^0 - 1\right| + e \cdot |-i f(t) + g(t)| \leq 3\left(\frac{\epsilon}{15} + \frac{\epsilon^2}{100}\right) \leq \frac{\epsilon}{3},
\]
since $\epsilon \leq \frac{1}{2}$.
Combining this with the above expression gives
\[
\begin{split}
    \Bigg|\int_{q(t) \leq \sigma} e^{-i \langle b, t \rangle} \phi_Y(t)\, \diff t &- \int_{q(t) \leq \sigma} e^{-q(t)} \diff t\Bigg| \leq \int_{q(t) \leq \sigma} e^{-q(t)} \left|e^{-i f(t) + g(t)} - 1\right| \diff t \\
        &= \int_{T \cap (\mathbb{R}^m \setminus B)} e^{-q(t)} \left|e^{-i f(t) + g(t)} - 1\right|\, \diff t + \int_{T \cap B} e^{-q(t)} \left|e^{-i f(t) + g(t)} - 1\right| \diff t \\
        &\leq \frac{3\epsilon}{16} \int_{\mathbb{R}^m} e^{-q(t)} \diff t + \frac{\epsilon}{3} \int_{T \cap B} e^{-q(t)} \diff t \\
        &\leq \frac{2\epsilon}{3} \int_{\mathbb{R}^m} e^{-q(t)} \diff t.
\end{split}
\]
By part $(1)$ of Lemma \ref{lem:q_bound_int}, we also have
\[
    \int_{q(t) \geq \sigma} e^{-q(t)} \diff t \leq e^{-\sigma/2} \int_{\mathbb{R}^m} e^{-q(t)} \diff t = e^{-2m} \epsilon^5 \int_{\mathbb{R}^m} e^{-q(t)} \diff t \leq \epsilon^5 \int_{\mathbb{R}^m} e^{-q(t)} \diff t.
\]
Combining these gives
\[
    \left|\int_{q(t) \leq \sigma} e^{-i \langle b, t \rangle} \phi_Y(t) \diff t - \int_{\mathbb{R}^m} e^{-q(t)} \diff t\right| \leq \left(\frac{2\epsilon}{3} + \epsilon^5\right) \int_{\mathbb{R}^m} e^{-q(t)} \diff t.
\]

\subsection{The characteristic function for large $t$} \label{sec:large_t}

\begin{lemma} \label{lem:char_func_bound}
    For all real $t \in \mathbb{R}^m$, the characteristic function of $Y$ can be bounded by
    \[
        |\phi_Y(t)| \leq \left(1 + \theta^2 \|t\|_2^2 \right)^{-\frac{\lambda}{\theta^2}}.
    \]
\end{lemma}
\begin{proof}
    Using the expressions in the proof of Lemma \ref{lem:char_func_expression_small_t} above, we have
    \[
        |\phi_Y(t)| = |\phi_Y(t)^2|^{\frac{1}{2}} = \left|\det\left[I - \frac{2i}{N+1} Z(t)\right]^{N+1}\right|^{-\frac{1}{2}} = \det\left(\left|I - \frac{2i}{N+1} Z(t)\right|^2\right)^{-\frac{N+1}{4}}.
    \]
    Note here that the absolute value takes care of the fact that we may need to take a square root of a complex number in the computation of $\phi_Y(t)$. With this, we then further compute
    \[
        |\phi_Y(t)| = \det\left(\left|I - \frac{2i}{N+1} Z(t)\right|^2\right)^{-\frac{N+1}{4}} = \det\left[I + \frac{4}{(N+1)^2} Z^2(t)\right]^{-\frac{N+1}{4}}.
    \]
    Now, denoting
    \[
        \xi_j := \frac{4}{(N+1)^2} \lambda_j^2(Z(t)),
    \]
    we have
    \[
        \sum_{j=1}^N \xi_j = \frac{4}{N+1} q(t) \geq \frac{4\lambda}{N+1} \|t\|_2^2
    \]
    by the assumption given in (\ref{eq:lambda_assumption}). By Lemma \ref{lem:Z_max_eig},
    \[
        \frac{4}{(N+1)^2} \lambda_{\max}^2(Z(t)) \leq \frac{4}{(N+1)^2} \left[\frac{N+1}{2} \theta \|t\|_2\right]^2 = \theta^2 \|t\|_2^2.
    \]
    We now want to minimize $\prod_{j=1}^N (1+\xi_j)$, a log-concave function, over the polytope given by
    \[
        \sum_{j=1}^N \xi_j \geq \frac{4\lambda}{N+1} \|t\|_2^2 \quad \text{and} \quad 0 \leq \xi_j \leq \theta^2\|t\|_2^2.
    \]
    By log-concavity, the minimum must occur at an extreme point of this polytope, thus at least $N$ inequalities above must actually be equalities. In particular, this means that all but at most one value of $\xi_j$ is equal to $0$ or $\theta^2\|t\|_2^2$. Given such a minimum, if $N_0$ is the number of values of $\xi_j$ that equal $\theta^2\|t\|_2^2$, then $N_0 \geq \left\lfloor \frac{4\lambda}{\theta^2(N+1)} \right\rfloor$. If $N_0 \geq \left\lceil \frac{4\lambda}{\theta^2(N+1)} \right\rceil$, then we have
    \[
        \prod_{j=1}^N (1+\xi_j) \geq \prod_{j=1}^{N_0} \left(1+\theta^2\|t\|_2^2\right) \geq \left(1+\theta^2\|t\|_2^2\right)^{\frac{4\lambda}{\theta^2(N+1)}}.
    \]
    Otherwise $N_0 = \left\lfloor \frac{4\lambda}{\theta^2(N+1)} \right\rfloor$. The one potentially non-extreme value of $\xi_j$, call it $\xi_{j_0}$, is then such that
    \[
        \frac{4\lambda}{N+1} \|t\|_2^2 \leq \sum_{j=1}^N \xi_j = \left\lfloor \frac{4\lambda}{\theta^2(N+1)} \right\rfloor \cdot \theta^2 \|t\|_2^2 + \xi_{j_0},
    \]
    which implies $\frac{\xi_{j_0}}{\theta^2\|t\|_2^2} \geq \frac{4\lambda}{\theta^2(N+1)} - \left\lfloor \frac{4\lambda}{\theta^2(N+1)} \right\rfloor =: \alpha_0$. Thus we have
    \[
        \prod_{j=1}^N (1+\xi_j) \geq \left(1+\theta^2\|t\|_2^2\right)^{\left\lfloor \frac{4\lambda}{\theta^2(N+1)} \right\rfloor} \left(1 + \alpha_0 \theta^2\|t\|_2^2 \right).
    \]
    Now for $\alpha \in [0,1]$ and $r > 0$, the function
    \[
        f(\alpha) := 1 + \alpha r - (1+r)^\alpha
    \]
    is concave, and $f(0) = f(1) = 0$. Therefore $f(\alpha) \geq 0$ for $\alpha \in [0,1]$, which in turn implies
    \[
        \prod_{j=1}^N (1+\xi_j) \geq \left(1+\theta^2\|t\|_2^2\right)^{\left\lfloor \frac{4\lambda}{\theta^2(N+1)} \right\rfloor + \alpha_0} = \left(1+\theta^2\|t\|_2^2\right)^{\frac{4\lambda}{\theta^2(N+1)}}.
    \]
    So in any case, the above inequality holds. Rearranging then finally gives
    \[
        |\phi_Y(t)| \leq \left[\left(1+\theta^2\|t\|_2^2\right)^{\frac{4\lambda}{\theta^2(N+1)}}\right]^{-\frac{N+1}{4}} = \left(1+\theta^2\|t\|_2^2\right)^{-\frac{\lambda}{\theta^2}}.
    \]
\end{proof}

\begin{proposition}
    The integral of $\phi_Y(t)$ over the region where $\|t\|_2 \geq \frac{1}{2\theta}$ can be bounded by
    \[
        \int_{t:~\|t\|_2 \geq \frac{1}{2\theta}} |\phi_Y(t)| \diff t \leq \frac{\epsilon((N+1)\pi)^{\frac{m}{2}}}{100 \sqrt{\det(BB^\top)}} = \frac{\epsilon}{100} \int_{\mathbb{R}^m} e^{-q(t)} \diff t.
    \]
\end{proposition}
\begin{proof}
    By Lemma \ref{lem:t_bound_int} and Lemma \ref{lem:char_func_bound} we have
    \[
    \begin{split}
        \int_{t:~\|t\|_2 \geq 1/2\theta} |\phi_Y(t)| \diff t &\leq \int_{t:~\|t\|_2 \geq 1/2\theta} (1+\theta^2\|t\|_2^2)^{-\lambda/\theta^2} \diff t \\
            &\leq \frac{2\pi^{\frac{m}{2}}}{\Gamma(\frac{m}{2}) \cdot \theta^m (\frac{2\lambda}{\theta^2} - m)} \left(\frac{5}{4}\right)^{\frac{m}{2} - \frac{\lambda}{\theta^2}} \\
            &= \frac{1}{\Gamma(\frac{m}{2}) \cdot (\frac{\lambda}{\theta^2} - \frac{m}{2})} \left(\frac{5\pi}{4\theta^2}\right)^{\frac{m}{2}} \left(\frac{5}{4}\right)^{-\frac{\lambda}{\theta^2}}.
    \end{split}
    \]
    By Lemma \ref{lem:q_B_expression} and Lemma \ref{lem:Z_max_eig} we also have
    \[
    \begin{split}
        \lambda_{\max}(BB^\top) &= \sup_{\|t\|_2 = 1} t^\top BB^\top t = \sup_{\|t\|_2 = 1} (N+1) q(t) = \sup_{\|t\|_2 = 1} \sum_{j=1}^N \lambda_j^2(Z(t)) \\
            &\leq \sup_{\|t\|_2 = 1} \sum_{j=1}^N \left(\frac{N+1}{2}\theta\|t\|_2\right)^2 = \frac{N(N+1)^2}{4} \theta^2,
    \end{split}
    \]
    which implies
    \[
        \frac{\det(BB^\top)}{(N+1)^m} \leq \left[\frac{N(N+1)^2}{4(N+1)} \theta^2\right]^m \leq \left[\frac{N^2\theta^2}{2}\right]^m.
    \]
    Now recall from the assumption given in (\ref{eq:gamma_assumption}) that
    \[
        \frac{\lambda}{\theta^2} \geq \gamma \epsilon^{-2} m\left(m + \log(\epsilon^{-1})\right)^2 \log(N\epsilon^{-1}).
    \]
    Since $\epsilon \leq \frac{1}{2}$, choosing $\gamma \geq 2$ implies
    \[
    \begin{split}
        \left(\frac{5}{4}\right)^{\frac{\lambda}{\theta^2}} &\geq \left(\frac{N}{\epsilon}\right)^{\gamma \epsilon^{-2} m (m+\log(\epsilon^{-1}))^2\log(\frac{5}{4})} \\
            &\geq (2N)^m = \left(\frac{8}{\theta^2}\right)^{m/2} \left(\frac{N^2\theta^2}{2}\right)^{m/2} \\
            &\geq \left(\frac{8}{\theta^2}\right)^{m/2} \left(\frac{\det(BB^\top)}{(N+1)^m}\right)^{1/2}
    \end{split}
    \]
    For $\gamma \geq 200$, we then further have
    \[
        \frac{\lambda}{\theta^2} - \frac{m}{2} \geq \frac{m}{2\epsilon^2} (2\gamma\log(2) - \epsilon^2) \geq \frac{100m}{\epsilon^2} \geq \frac{200m}{\epsilon}.
    \]
    which implies
    \[
        \Gamma\left(\frac{m}{2}\right) \cdot \left(\frac{\lambda}{\theta^2} - \frac{m}{2}\right) \geq \Gamma\left(\frac{m}{2}\right) \cdot \frac{200m}{\epsilon} \geq \frac{100m}{\epsilon},
    \]
    since $\Gamma(x) \geq \frac{1}{2}$ for $x > 0$. Combining everything then gives
    \[
    \begin{split}
        \int_{t:~\|t\|_2 \geq 1/2\theta} |\phi_Y(t)| \diff t &\leq \frac{1}{\Gamma(\frac{m}{2}) \cdot (\frac{\lambda}{\theta^2} - \frac{m}{2})} \left(\frac{5\pi}{4\theta^2}\right)^{m/2} \left(\frac{5}{4}\right)^{-\lambda/\theta^2} \\
            &\leq \frac{\epsilon}{100m} \left(\frac{5\pi}{4\theta^2}\right)^{m/2} \left(\frac{\theta^2}{8}\right)^{m/2} \left(\frac{(N+1)^m}{\det(BB^\top)}\right)^{1/2} \\
            &\leq \frac{\epsilon}{100} \left(\frac{5}{4 \cdot 8}\right)^{m/2} \left(\frac{(\pi(N+1))^m}{\det(BB^\top)}\right)^{1/2} \\
            &\leq \frac{\epsilon}{100} \left(\frac{(\pi(N+1))^m}{\det(BB^\top)}\right)^{1/2}.
    \end{split}
    \]
    This finishes the proof of the first claim. For the second claim, suppose that $\mu_1,\mu_2,\dots,\mu_m\geq 0$ are the eigenvalues of $BB^T$. Then, by diagonalizing $BB^T$ with an orthogonal matrix \[
    \begin{split}
    \int_{\R^m} e^{-q(t)}\, \diff t &= \int_{\R^m} e^{-(\sum_{i=1}^m \mu_i t_i^2) / (N+1)}\,\diff t\\
    &= \prod_{i=1}^m \int_{\R} e^{-\mu_i t^2 / (N+1)}\, \diff t\\
    &=\prod_{i=1}^m \frac{\sqrt{\pi(N+1)}}{\sqrt{\mu_i}}\\
    &=\frac{(\pi(N+1))^{m/2}}{\det(BB^T)^{1/2}}
    \end{split}
    \] holds. This finishes the proof. 
    
\end{proof}

\section{Proofs of the Main Results} \label{sec:proofs_corollaries}

In this section, we prove the main results of Section \ref{sec:main_result} as corollaries of Theorem \ref{thm:main}. We first recall the spectrahedron notation given in Section \ref{sec:main_result}, and then we recall the conditions of Theorem \ref{thm:main} given in Section \ref{sec:proof_main}.

Given $A_1,A_2,\dots,A_m\in\Sym(N)$ and $b\in\R^m$, we define a spectrahedron $\mathcal{S}$ by:
\[
    \mathcal{S} := \big\{P \in \PSD(N) \; : \; \tr(A_k P) = b_k \quad \text{for} \quad k \in [m]\big\}.
\]
We assume that $\mathcal{S}$ is compact, that the constraints $\tr(A_k P) = b_k$ are linearly independent, that $m < \binom{N+1}{2} = \dim(\PSD(N))$, and that $\mathcal{S}$ is of dimension exactly $\binom{N+1}{2} - m$. Let $P^\star \in \mathcal{S}$ be the point which maximizes the function 
\[
\begin{split}
    \phi(P) &= \log\Gamma_N\left(\frac{N+1}{2}\right) - \frac{N(N+1)}{2} \log\left(\frac{N+1}{2e}\right) + \frac{N+1}{2}\log\det(P) \\
        &= \text{const}(N) + \frac{N+1}{2}\log\det(P)
\end{split}
\]
over $\mathcal{S}$.
We let $A$ and $B$ be linear operators from $\Sym(N)$ to $\R^m$, defined via
\[
    AX := (\tr(A_1X), \ldots, \tr(A_mX))
\]
and
\[
    BX := (\tr(Z_1X),\ldots,\tr(Z_mX)).
\]
where $Z_k := \sqrt{P^\star} A_k \sqrt{P^\star}$ for all $k \in [m]$.

In the statement of Theorem \ref{thm:main}, we consider the quadratic form $q: \R^m \to \R$ defined by
\[
    q(t) := \frac{1}{N+1} \tr\left[\left(\sum_{k=1}^m t_k Z_k\right)^2\right].
\]
We suppose that for some $\lambda > 0$ we have
\begin{align} \label{eq:main1_2} \tag{A1}
    q(t) \geq \lambda \|t\|_2^2 \quad \text{for all} \quad t \in \R^m,
\end{align}
and that for some $\theta > 0$ we have
\begin{align} \label{eq:main2_2} \tag{A2}
    \frac{2}{N+1} \|(x^\top Z_k x)_{k=1}^m\|_2 \leq \theta \quad \text{for all} \quad \|x\|_2 = 1.
\end{align}
Given $0 < \epsilon \leq \frac{1}{2}$, we further suppose that
\begin{align} \label{eq:main3_2} \tag{A3}
    \lambda \geq \gamma \theta^2 \epsilon^{-2} m \left(m + \log(\epsilon^{-1})\right)^2 \log(N\epsilon^{-1}),
\end{align}
where $\gamma = 10^5$ is an absolute constant.
Under these conditions, Theorem \ref{thm:main} gives an approximate volume formula.

\subsection{Simplifying Theorem \ref{thm:main}} \label{sec:simplify_main}

We now demonstrate a way to simplify the above conditions of Theorem \ref{thm:main}, which we will use in the proofs of the main results.

By Lemma \ref{lem:q_B_expression}, we have that $q(t) = \frac{1}{N+1} t^\top BB^\top t$, and so $\lambda$ in Condition \ref{eq:main1_2} can be optimally chosen to be the minimal eigenvalue of $\frac{BB^\top}{N+1}$.
Since $B$ is a linear map taking matrices as input, we may write
\[
    (x^\top Z_k x)_{k=1}^m = \left(\tr(\sqrt{P^\star} A_k \sqrt{P^\star} xx^\top)\right)_{k=1}^m = B(xx^\top).
\]
Since $\|x\|_2 = 1$ if and only if $\|xx^\top\|_F = 1$, where $\|X\|_F = \tr(X^\top X)$ denotes the entrywise 2-norm of $X$, we can choose $\theta$ to be the maximal eigenvalue of $\frac{2 \sqrt{BB^\top}}{N+1}$ (though this is non-optimal).
With this, we can replace Condition (\ref{eq:main3_2}) above by
\[
    \frac{\lambda_{\min}(BB^\top)}{\lambda_{\max}(BB^\top)} \geq \frac{4 \gamma}{\epsilon^2 (N+1)} m \left(m + \log(\epsilon^{-1})\right)^2 \log(N\epsilon^{-1}),
\]
where $\lambda_{\min}$ and $\lambda_{\max}$ refer to the minimum and maximum eigenvalues respectively.

That is, we can replace Condition (\ref{eq:main3_2}) by a bound on the condition number of $BB^\top$.
What is special about this is the fact that one can always change $A_k$ and $b_k$ defined above to enforce the condition number of $BB^\top$ to be 1, without changing the spectrahedron or the value of $P^\star$.
We prove this formally in Lemma \ref{lem:spectra_repr} below.
Thus we can further replace Condition (\ref{eq:main3_2}) above by
\begin{align} \label{eq:main3_2_new} \tag{A3\ensuremath{'}}
    \epsilon^2 (N+1) \geq 4 \gamma m \left(m + \log(\epsilon^{-1})\right)^2 \log(N\epsilon^{-1}),
\end{align}
That is, to obtain the volume approximation for a given $\epsilon$, it is enough to achieve the above bound comparing the size $N$ of the matrices under consideration to the number $m$ of linear constraints on the spectrahedron.

We now prove formally that we can assume the condition number of $BB^\top$ to be 1.


\begin{lemma} \label{lem:spectra_repr}
    Let $\mathcal{S}$ be a spectrahedron as defined above, so that
    \[
        \mathcal{S} := \left\{P \in \PSD(N) : \tr(A_k P) = b_k, ~ k \in [m]\right\},
    \]
    where $A_k$ are linearly independent for $k \in [m]$.
    There exists a choice of $A_k'$ and $b_k'$ which defines the same spectrahedron $\mathcal{S}$, so that $B'(B')^\top = I_m$ (where $B'$ is defined analogously with respect to $A'$ as $B$ is to $A$).
    Further, the matrix $P^\star$ and the volume approximation formula for $\mathcal{S}$ are unchanged by replacing $A_k$ and $b_k$ by $A_k'$ and $b_k'$ respectively.
\end{lemma}
\begin{proof}
    First note that since the entropy function $\phi$ given above is not dependent on the representation of $\mathcal{S}$, we have that the optimizer $P^\star$ is also not dependent on the representation of $\mathcal{S}$.
    Now, $\mathcal{S}$ can be defined as the set of all positive semi-definite matrices satisfying the linear system given by
    \[
        b = AX = \begin{bmatrix}
            | & | &  & | \\
            \mathrm{vec}(A_1) & \mathrm{vec}(A_2) & \cdots & \mathrm{vec}(A_m) \\
            | & | &  & | \\
        \end{bmatrix}^\top \mathrm{vec}(X).
    \]
    Using the above notation, we defined $Z_k := \sqrt{P^\star} A_k \sqrt{P^\star}$.
    Considering the invertible linear map $\mathcal{L}_{P^\star} : A \mapsto \sqrt{P^\star} A \sqrt{P^\star}$ on real symmetric matrices, we can write
    \[
        Z_k = \sqrt{P^\star} A_k \sqrt{P^\star} = \mathrm{vec}^{-1}(L_{P^\star} \cdot \mathrm{vec}(A_k))
    \]
    for some $n^2 \times n^2$ invertible matrix $L_{P^\star}$ which represents the linear map $\mathcal{L}_{P^\star}$.
    With this, we can write
    \[
        BX = \begin{bmatrix}
            | & | &  & | \\
            \mathrm{vec}(Z_1) & \mathrm{vec}(Z_2) & \cdots & \mathrm{vec}(Z_m) \\
            | & | &  & | \\
        \end{bmatrix}^\top \mathrm{vec}(X)
        = \left(A \cdot L_{P^\star}^\top\right) \mathrm{vec}(X).
    \]
    Now let $B = U \Sigma V^\top$ be the real singular value decomposition of $B$, where $\Sigma$ is an $m \times N^2$ rectangular diagonal matrix with non-negative entries and $U$ and $V$ are real orthogonal matrices of appropriate size.
    Further, since $m < N^2$ and the matrices $A_k$ are linearly independent for $k \in [m]$, we have that
    \[
        \Sigma = \begin{bmatrix} D & 0_{N^2-m} \end{bmatrix},
    \]
    where $D$ is an $m \times m$ diagonal matrix with strictly positive entries.
    Now define $C := D^{-1} U^{-1}$ and
    \[
        A' := CA = CB\left(L_{P^\star}^\top\right)^{-1} \quad \text{and} \quad b' = Cb.
    \]
    Note that since $A$ is a matrix whose rows are vectorizations of real symmetric matrices, we have that $A'$ is also a matrix of rows are vectorizations of real symmetric matrices.
    Therefore $A'$ is the linear system corresponding to another spectrahedron given by
    \[
        \mathcal{S}' := \left\{P \in \PSD(N) : \tr(A_k' P) = b_k', ~ k \in [m]\right\},
    \]
    where the matrices $A_k'$ correspond to the rows of $A'$.
    Since $C$ is invertible, we have in fact that $\mathcal{S}' = \mathcal{S}$, and therefore
    \[
        B' = A' \cdot L_{P^\star}^\top = \left(CB\left(L_{P^\star}^\top\right)^{-1}\right) L_{P^\star}^\top = CB = D^{-1} U^{-1} U \Sigma V^\top = \begin{bmatrix} I_m & 0_{N^2-m} \end{bmatrix} V^\top,
    \]
    where $B'$ is defined analogously with respect to $A'$ as $B$ is to $A$.
    Thus $BB^\top = I_m$.
    Finally, since
    \[
        \frac{\det(A'(A')^\top)}{\det(B'(B')^\top)} = \frac{\det(CAA^\top C^\top)}{\det(CBB^\top C^\top)} = \frac{\det(AA^\top)}{\det(BB^\top)},
    \]
    we have that the volume approximation formula remains unchanged by replacing $A_k$ and $b_k$ by $A_k'$ and $b_k'$ respectively.
\end{proof}

With this lemma in hand, we can now prove the main results.

\subsection{Proof of Theorem \ref{thm:approx}: Main approximation result}

We want to apply Theorem \ref{thm:main}, where we replace Condition \ref{eq:gamma_assumption} by Condition \ref{eq:main3_2_new} via the discussion of Section \ref{sec:simplify_main}. To this end, we first define
\[
    \delta := 32\gamma \cdot \frac{m^3 \log N}{N},
\]
where $\gamma = 10^5$ as in Theorem \ref{thm:main}. Recall the assumptions of Theorem \ref{thm:approx}: $\epsilon \leq e^{-1}$ and Condition (\ref{eq:main4}), which states
\[
    \frac{\epsilon^2}{\log^3(\epsilon^{-1})} \geq 32 \gamma \cdot \frac{m^3 \log N}{N}.
\]
Note that the extra factor of $32$ here is due to the difference in the value of the constant $\gamma$. Thus we have
\[
    \epsilon^2 \geq \delta \cdot \log^3(\epsilon^{-1}) \geq \frac{\delta}{8} (1+\log(\epsilon^{-1}))^3 \geq \frac{\delta}{8} \left(1+\frac{\log(\epsilon^{-1})}{m}\right)^2 \left(1+\frac{\log(\epsilon^{-1})}{\log N}\right).
\]
Therefore
\[
    \epsilon^2 \geq \frac{4\gamma m (m + \log(\epsilon^{-1}))^2 (\log N + \log(\epsilon^{-1}))}{N} \geq \frac{4\gamma m (m + \log(\epsilon^{-1}))^2 \log(N\epsilon^{-1})}{N+1},
\]
which is precisely Condition \ref{eq:main3_2_new}. This completes the proof.


\subsection{Proof of Theorem \ref{thm:asymptotic}: Main asymptotic result}

%
Recall the assumptions of Theorem \ref{thm:asymptotic}: $\epsilon < e^{-1}$ and Condition (\ref{eq:main5}), which states
\[
    \lim_{n \to \infty} \frac{m_n^3 \log N_n}{N_n} = 0.
\]
Now recall that $m_n$ and $N_n$ are positive integers.
By the assumption that $m_n < N_n^2$, this implies $N_n \geq 2$ and
\[
    \lim_{n \to \infty} \frac{4\gamma m_n \left(m_n+\log(\epsilon^{-1})\right)^2 \log(N_n\epsilon^{-1})}{\epsilon^2(N_n+1)} \leq \frac{8\gamma \left(1+\log(\epsilon^{-1})\right)^3}{\epsilon^2} \lim_{n \to \infty} \frac{m_n^3 \log N_n}{N_n} = 0.
\]
%
That is, there exists $n_\epsilon$ such that for all $n \geq n_\epsilon$ we have that condition (\ref{eq:main3_2_new}) in Section \ref{sec:simplify_main} is satisfied for $\mathcal{S}_n$.
Therefore by Theorem \ref{thm:approx} and Section \ref{sec:simplify_main}, we have
\[
    \left|\frac{\vol(\mathcal{S}_n)}{\left(\frac{N_n+1}{4\pi}\right)^{m_n/2} \left(\frac{\det(A_nA_n^\top)}{\det(B_nB_n^\top)}\right)^{1/2} e^{\phi_n(P_n^\star)}} - 1\right| \leq \epsilon
\]
for all $n \geq n_\epsilon$.
This completes the proof.

\subsection{Proof of Corollary \ref{cor:sections}: Central sections of the spectraplex}

Recall in Corollary~\ref{cor:sections} that we define
\[
\mathcal{S}_1 := \left\{ P\in\PSD(N)\, : \, \tr(P)=1 \right\}
\]
as the standard spectraplex, and for $M\in\Sym(N)$ that is linearly independent of $I_N$ we define
\[
\mathcal{S}_M := \left\{ P\in\PSD(N)\, :\, \tr(P)=1\,\text{ and }\, \tr(PM)=\frac{1}{N}\tr(M) \right\},
\]
which we call a central section of $\mathcal{S}_1$.

\begin{lemma}
For any $M\in\Sym(N)$ that is not a scalar multiple of the identity matrix, $\frac{1}{N}I_N$ maximizes $\phi$ on $\mathcal{S}_1$ and on $\mathcal{S}_M$.
\end{lemma}
\begin{proof}
    The gradient of $\phi$ satisfies $\nabla\phi(P)=\frac{N+1}{2}P^{-1}$. Since $\mathcal{S}_1$ orthogonal to the line spanned by $\frac{1}{N} I_N$, we have that a matrix $P^*$ maximizes $\phi$ on $\mathcal{S}_1$ if and only if $(P^*)^{-1}$ (and thus $P^*$) is a scalar multiple of $I_N$. This shows that $\frac{1}{N}I_N$ maximizes $\phi$ over $\mathcal{S}_1$. Since $\frac{1}{N} I_N \in \mathcal{S}_M \subset \mathcal{S}_1$, it also maximizes $\phi$ over $\mathcal{S}_M$.
\end{proof}

With this, we prove Corollary \ref{cor:sections} as follows. Let $\{\mathcal{S}_{M_n}\}_{n=1}^\infty$ be any sequence of central sections of $\mathcal{S}_1$, with $N_n \to \infty$. Thus we have
\[
    \lim_{n \to \infty} \frac{m_n^3 \log N_n}{N_n} = \lim_{n \to \infty} \frac{8 \log N_n}{N_n} = 0.
\]
Thus for every $\epsilon \leq e^{-1}$, there is some $n_\epsilon$ such that for $n \geq n_\epsilon$ we have
\[
    \frac{\epsilon^2}{\log^3(\epsilon^{-1})} \geq 32\gamma \cdot \frac{m_n^3 \log N_n}{N_n},
\]
which is the condition required to apply Theorem \ref{thm:approx}. Applying Theorem \ref{thm:approx} then completes the proof. For the asymptotic statement at the end of Corollary \ref{cor:sections}, one can also directly apply Theorem \ref{thm:asymptotic}.

As a final comment, note that this result says: for large enough $N$, all central sections of the spectraplex have volume close to the expected volume over all central sections. See \cite{randomspectrahedra} for further discussion on the volume of random spectrahedra.

\small{
\paragraph{Acknowledgements.} The authors would like to thank Alexander Barvinok, Peter B\"urgisser, and Akshay Ramachandran for very helpful discussions. The first author is supported by the ERC under the European's Horizon 2020 research and innovation programme (grant 787840).  The second author was partially supported by the Deutsche Forschungsgemeinschaft (DFG, German Research Foundation) under Germany’s Excellence Strategy – The Berlin Mathematics Research Center MATH+ (EXC-2046/1, project ID:
390685689). The third author would like to thank the Institute of Mathematical Sciences, Chennai, for hosting him in Spring 2022, when part of this work was done.  He also gratefully acknowledges financial support from the Bogazici University Solidarity fund.
}

\bibliographystyle{alpha}
\bibliography{references}

\newcommand{\etalchar}[1]{$^{#1}$}
\begin{thebibliography}{CDWY18}

\bibitem[Ali95]{alizadeh}
Farid Alizadeh.
\newblock Interior point methods in semidefinite programming with applications
  to combinatorial optimization.
\newblock {\em SIAM Journal on Optimization}, 5(1):13--51, 1995.

\bibitem[AS17]{aubrun2017alice}
Guillaume Aubrun and Stanis{\l}aw~J Szarek.
\newblock {\em Alice and Bob meet Banach}, volume 223.
\newblock American Mathematical Soc., 2017.

\bibitem[BFG{\etalchar{+}}18]{burgisser2018efficient}
Peter B{\"u}rgisser, Cole Franks, Ankit Garg, Rafael Oliveira, Michael Walter,
  and Avi Wigderson.
\newblock Efficient algorithms for tensor scaling, quantum marginals, and
  moment polytopes.
\newblock In {\em 2018 IEEE 59th Annual Symposium on Foundations of Computer
  Science (FOCS)}, pages 883--897. IEEE, 2018.

\bibitem[BGO{\etalchar{+}}18]{altmin}
Peter Bürgisser, Ankit Garg, Rafael Oliveira, Michael Walter, and Avi
  Wigderson.
\newblock Alternating minimization, scaling algorithms, and the null-cone
  problem from invariant theory.
\newblock 2018.

\bibitem[BH93]{betke1993approximating}
Ulrich Betke and Martin Henk.
\newblock Approximating the volume of convex bodies.
\newblock {\em Discrete \& Computational Geometry}, 10(1):15--21, 1993.

\bibitem[BH09]{barvinok2009maximum}
Alexander Barvinok and JA~Hartigan.
\newblock Maximum entropy {E}dgeworth estimates of the number of integer points
  in polytopes.
\newblock {\em arXiv preprint arXiv:0910.2497}, 2009.

\bibitem[BH10]{bhgaussian}
Alexander Barvinok and JA~Hartigan.
\newblock Maximum entropy gaussian approximations for the number of integer
  points and volumes of polytopes.
\newblock {\em Advances in Applied Mathematics}, 45(2):252--289, 2010.

\bibitem[BH12]{barvinok2012asymptotic}
Alexander Barvinok and JA~Hartigan.
\newblock An asymptotic formula for the number of non-negative integer matrices
  with prescribed row and column sums.
\newblock {\em Transactions of the American Mathematical Society},
  364(8):4323--4368, 2012.

\bibitem[Bil95]{billingsley}
Patrick Billingsley.
\newblock {\em Probability and Measure}.
\newblock Wiley series in probability and mathematical statistics. Wiley, New
  York u.a., 3. ed. edition, 1995.

\bibitem[BKL19]{randomspectrahedra}
Paul Breiding, Khazhgali Kozhasov, and Antonio Lerario.
\newblock Random spectrahedra.
\newblock {\em SIAM Journal on Optimization}, 29(4):2608--2624, 2019.

\bibitem[BP14]{benson2014counting}
David Benson-Putnins.
\newblock Counting integer points in multi-index transportation polytopes.
\newblock {\em arXiv preprint arXiv:1402.4715}, 2014.

\bibitem[BPT12]{blekherman2012semidefinite}
Grigoriy Blekherman, Pablo~A Parrilo, and Rekha~R Thomas.
\newblock {\em Semidefinite optimization and convex algebraic geometry}.
\newblock SIAM, 2012.

\bibitem[BR21]{barvinok2021quick}
Alexander Barvinok and Mark Rudelson.
\newblock A quick estimate for the volume of a polyhedron.
\newblock {\em arXiv preprint arXiv:2112.06322}, 2021.

\bibitem[BR22]{barvinok2022system}
Alexander Barvinok and Mark Rudelson.
\newblock When a system of real quadratic equations has a solution.
\newblock {\em Advances in Mathematics}, 403:108391, 2022.

\bibitem[Brz13]{brzezinski2013volume}
Patryk Brzezinski.
\newblock Volume estimates for sections of certain convex bodies.
\newblock {\em Mathematische Nachrichten}, 286(17-18):1726--1743, 2013.

\bibitem[BV04]{boyd2004convex}
Stephen Boyd and Lieven Vandenberghe.
\newblock {\em Convex optimization}.
\newblock Cambridge university press, 2004.

\bibitem[CDWY18]{chen2018fast}
Yuansi Chen, Raaz Dwivedi, Martin~J Wainwright, and Bin Yu.
\newblock Fast {MCMC} sampling algorithms on polytopes.
\newblock {\em The Journal of Machine Learning Research}, 19(1):2146--2231,
  2018.

\bibitem[CEF19]{chalkis2019practical}
Apostolos Chalkis, Ioannis~Z Emiris, and Vissarion Fisikopoulos.
\newblock Practical volume estimation by a new annealing schedule for cooling
  convex bodies.
\newblock {\em arXiv preprint arXiv:1905.05494}, 2019.

\bibitem[CF20]{chalkis2020volesti}
Apostolos Chalkis and Vissarion Fisikopoulos.
\newblock volesti: {V}olume approximation and sampling for convex polytopes in
  {R}.
\newblock {\em arXiv preprint arXiv:2007.01578}, 2020.

\bibitem[CFRT21]{elias}
Apostolos Chalkis, Vissarion Fisikopoulos, Panagiotis Repouskos, and Elias
  Tsigaridas.
\newblock Sampling the feasible sets of {SDP}s and volume approximation.
\newblock {\em ACM Commun. Comput. Algebra}, 54(3):114–118, mar 2021.

\bibitem[Cho75]{choi1975completely}
Man-Duen Choi.
\newblock Completely positive linear maps on complex matrices.
\newblock {\em Linear algebra and its applications}, 10(3):285--290, 1975.

\bibitem[CM07]{canfield2007asymptotic}
E~Rodney Canfield and Brendan~D McKay.
\newblock The asymptotic volume of the {B}irkhoff polytope.
\newblock {\em arXiv preprint arXiv:0705.2422}, 2007.

\bibitem[Cou17]{cousins2017efficient}
Benjamin Cousins.
\newblock {\em Efficient high-dimensional sampling and integration}.
\newblock PhD thesis, Georgia Institute of Technology, 2017.

\bibitem[CV16]{cousins2016practical}
Ben Cousins and Santosh Vempala.
\newblock A practical volume algorithm.
\newblock {\em Mathematical Programming Computation}, 8(2):133--160, 2016.

\bibitem[DF88]{dyerfrieze}
M.~E. Dyer and A.~M. Frieze.
\newblock On the complexity of computing the volume of a polyhedron.
\newblock {\em SIAM Journal on Computing}, 17(5):967--974, 1988.

\bibitem[DFK91]{kannanetal}
Martin Dyer, Alan Frieze, and Ravi Kannan.
\newblock A random polynomial-time algorithm for approximating the volume of
  convex bodies.
\newblock {\em J. ACM}, 38(1):1–17, jan 1991.

\bibitem[EF14]{emiris2014efficient}
Ioannis~Z Emiris and Vissarion Fisikopoulos.
\newblock Efficient random-walk methods for approximating polytope volume.
\newblock In {\em Proceedings of the thirtieth annual symposium on
  Computational geometry}, pages 318--327, 2014.

\bibitem[EF18]{emiris2018practical}
Ioannis~Z Emiris and Vissarion Fisikopoulos.
\newblock Practical polytope volume approximation.
\newblock {\em ACM Transactions on Mathematical Software (TOMS)}, 44(4):1--21,
  2018.

\bibitem[Ele86]{elekes_geometric_1986}
G.~Elekes.
\newblock A geometric inequality and the complexity of computing volume.
\newblock {\em Discrete \& Computational Geometry}, 1(4):289--292, December
  1986.

\bibitem[G{\"u}l96]{guler1996barrier}
Osman G{\"u}ler.
\newblock Barrier functions in interior point methods.
\newblock {\em Mathematics of Operations Research}, 21(4):860--885, 1996.

\bibitem[Jay57a]{jaynes1957informationi}
Edwin~T Jaynes.
\newblock Information theory and statistical mechanics {I}.
\newblock {\em Physical review}, 106(4):620, 1957.

\bibitem[Jay57b]{jaynes1957informationii}
Edwin~T Jaynes.
\newblock Information theory and statistical mechanics {II}.
\newblock {\em Physical review}, 108(2):171, 1957.

\bibitem[KLS97]{kannan1997random}
Ravi Kannan, L{\'a}szl{\'o} Lov{\'a}sz, and Mikl{\'o}s Simonovits.
\newblock Random walks and an ${O}^*(n^5)$ volume algorithm for convex bodies.
\newblock {\em Random Structures \& Algorithms}, 11(1):1--50, 1997.

\bibitem[Kly02]{kly1}
Alexander Klyachko.
\newblock Coherent states, entanglement, and geometric invariant theory.
\newblock 2002.

\bibitem[Kly04]{kly2}
Alexander Klyachko.
\newblock Quantum marginal problem and representations of the symmetric group.
\newblock 2004.

\bibitem[LS93a]{landau1993birkhoff}
LJ~Landau and RF~Streater.
\newblock On {B}irkhoff's theorem for doubly stochastic completely positive
  maps of matrix algebras.
\newblock {\em Linear algebra and its applications}, 193:107--127, 1993.

\bibitem[LS93b]{lovasz1993random}
L{\'a}szl{\'o} Lov{\'a}sz and Mikl{\'o}s Simonovits.
\newblock Random walks in a convex body and an improved volume algorithm.
\newblock {\em Random structures \& algorithms}, 4(4):359--412, 1993.

\bibitem[LV06]{lovasz2006simulated}
L{\'a}szl{\'o} Lov{\'a}sz and Santosh Vempala.
\newblock Simulated annealing in convex bodies and an ${O}^*(n^4)$ volume
  algorithm.
\newblock {\em Journal of Computer and System Sciences}, 72(2):392--417, 2006.

\bibitem[LV18]{lee2018convergence}
Yin~Tat Lee and Santosh~S Vempala.
\newblock Convergence rate of {R}iemannian {H}amiltonian {M}onte {C}arlo and
  faster polytope volume computation.
\newblock In {\em Proceedings of the 50th Annual ACM SIGACT Symposium on Theory
  of Computing}, pages 1115--1121, 2018.

\bibitem[LV20]{leake2020computability}
Jonathan Leake and Nisheeth~K Vishnoi.
\newblock On the computability of continuous maximum entropy distributions with
  applications.
\newblock In {\em Proceedings of the 52nd Annual ACM SIGACT Symposium on Theory
  of Computing}, pages 930--943, 2020.

\bibitem[MV19]{mangoubi2019faster}
Oren Mangoubi and Nisheeth~K Vishnoi.
\newblock Faster polytope rounding, sampling, and volume computation via a
  sub-linear ball walk.
\newblock In {\em 2019 IEEE 60th Annual Symposium on Foundations of Computer
  Science (FOCS)}, pages 1338--1357. IEEE, 2019.

\bibitem[NN94]{nn}
Yurii Nesterov and Arkadii Nemirovskii.
\newblock {\em Interior-Point Polynomial Algorithms in Convex Programming}.
\newblock Society for Industrial and Applied Mathematics, 1994.

\bibitem[NY76]{yudinellipsoid}
Arkadi~S Nemirovskii and David~Berkovich Yudin.
\newblock Informational complexity and efficient methods for the solution of
  convex extremal problems.
\newblock {\em Matekon}, 13(2):22--45, 1976.

\bibitem[NY77]{nemirovskiellipsoid}
Arkadi~S Nemirovski and David~Berkovich Yudin.
\newblock Optimization methods adaptive to significant dimension of the
  problem.
\newblock {\em Avtomatika i Telemekhanika}, (4):75--87, 1977.

\bibitem[Ren01]{renegar}
James Renegar.
\newblock {\em A Mathematical View of Interior-Point Methods in Convex
  Optimization}.
\newblock Society for Industrial and Applied Mathematics, 2001.

\bibitem[Sha48]{shannon1948mathematical}
Claude~Elwood Shannon.
\newblock A mathematical theory of communication.
\newblock {\em The Bell system technical journal}, 27(3):379--423, 1948.

\bibitem[Sho77]{shorellipsoid}
Naum~Z Shor.
\newblock Cut-off method with space extension in convex programming problems.
\newblock {\em Cybernetics}, 13(1):94--96, 1977.

\bibitem[SV14]{singh2014entropy}
Mohit Singh and Nisheeth~K Vishnoi.
\newblock Entropy, optimization and counting.
\newblock In {\em Proceedings of the forty-sixth annual ACM symposium on Theory
  of computing}, pages 50--59, 2014.

\bibitem[Tod01]{todd_2001}
MJ~Todd.
\newblock Semidefinite optimization.
\newblock {\em Acta Numerica}, 10:515–560, 2001.

\bibitem[VBW98]{maxdet}
Lieven Vandenberghe, Stephen Boyd, and Shao-Po Wu.
\newblock Determinant maximization with linear matrix inequality constraints.
\newblock {\em SIAM Journal on Matrix Analysis and Applications},
  19(2):499--533, 1998.

\bibitem[Web96]{webb}
Simon Webb.
\newblock Central slices of the regular simplex.
\newblock {\em Geometriae Dedicata}, 61(1):19--28, June 1996.

\bibitem[WSV12]{wolkowicz2012handbook}
Henry Wolkowicz, Romesh Saigal, and Lieven Vandenberghe.
\newblock {\em Handbook of semidefinite programming: {T}heory, algorithms, and
  applications}, volume~27.
\newblock Springer Science \& Business Media, 2012.

\end{thebibliography}

\end{document}